\documentclass[11pt]{article}

\usepackage[utf8]{inputenc}
\usepackage{hyperref}
\usepackage{amsmath,amsthm,amssymb,caption,subcaption,nicefrac}
\usepackage[margin=1in]{geometry}
\usepackage{bm}
\usepackage[table]{xcolor}
\usepackage{multirow}
\usepackage{arydshln}
\usepackage{csquotes}
\usepackage{anyfontsize}
\usepackage[shortlabels]{enumitem}

\usepackage{amsmath}
\usepackage{amssymb}
\usepackage{mathtools}
\usepackage{amsthm}

\usepackage{algorithm}
\usepackage[noend]{algorithmic}
\usepackage[algo2e,ruled,noend]{algorithm2e}

\usepackage{cleveref}

\newtheorem{theorem}{Theorem}

\newtheorem{lemma}[theorem]{Lemma}
\newtheorem{corollary}[theorem]{Corollary}

\newtheorem{definition}[theorem]{Definition}
\newtheorem{observation}[theorem]{Observation}

\newtheorem{assumption}[theorem]{Assumption}
\newtheorem{claim}[theorem]{Claim}

\usepackage{etoolbox}
\AtBeginEnvironment{proposition}{\crefalias{theorem}{proposition}}
\AtBeginEnvironment{lemma}{\crefalias{theorem}{lemma}}
\AtBeginEnvironment{corollary}{\crefalias{theorem}{corollary}}
\AtBeginEnvironment{fact}{\crefalias{theorem}{fact}}
\AtBeginEnvironment{definition}{\crefalias{theorem}{definition}}
\AtBeginEnvironment{observation}{\crefalias{theorem}{observation}}
\AtBeginEnvironment{conjecture}{\crefalias{theorem}{conjecture}}
\AtBeginEnvironment{assumption}{\crefalias{theorem}{assumption}}
\AtBeginEnvironment{claim}{\crefalias{theorem}{claim}}

\newcommand{\cX}{\mathcal{X}}

\newcommand{\cM}{\mathcal{M}}

\newcommand{\ind}[1]{\mathbf{1}\left\{#1\right\}}
\newcommand{\eps}{\varepsilon}
\newcommand{\N}{\mathbb{N}}
\newcommand{\Z}{\mathbb{Z}}
\newcommand{\R}{\mathbb{R}}

\newcommand{\E}{\mathbb{E}}
\newcommand{\cE}{\mathcal{E}}
\newcommand{\cG}{\mathcal{G}}

\newcommand{\cD}{\mathcal{D}}
\newcommand{\cP}{\mathcal{P}}

\newcommand{\cH}{\mathcal{H}}

\renewcommand{\phi}{\varphi}

\newcommand{\tc}{\tilde{c}}
\newcommand{\tT}{\tilde{T}}

\newcommand{\tO}{\tilde{O}}

\newcommand{\cC}{\mathcal{C}}

\newcommand{\cO}{\mathcal{O}}

\DeclareMathOperator{\Geo}{Geom}
\DeclareMathOperator{\Ber}{Ber}
\DeclareMathOperator{\Lap}{Lap}

\DeclareMathOperator{\dist}{dist}

\newcommand{\alg}{\textsc{ALG}}
\newcommand{\algbaseline}{\textsc{ALG}_{\mathrm{base}}}
\newcommand{\algdeletion}{\textsc{ALG}_{\mathrm{del}}}

\newcommand{\EM}{\textsc{ExpMech}}

\newcommand{\opt}{\textsc{opt}}
\newcommand{\val}{\textsc{val}}

\renewcommand{\setminus}{\smallsetminus}

\newcommand{\scr}{\mathrm{scr}}

\newcommand{\Edef}{\mathcal{E}_{\mathrm{def}}}
\newcommand{\Egood}{\mathcal{E}_{\mathrm{good}}}

\newcommand{\Paren}[1]{\left(#1\right)}

\allowdisplaybreaks

\definecolor{Gred}{RGB}{219, 50, 54}
\definecolor{Ggreen}{RGB}{60, 186, 84}
\definecolor{Gblue}{RGB}{72, 133, 237}
\definecolor{Gyellow}{RGB}{247, 178, 16}
\definecolor{ToCgreen}{RGB}{0, 128, 0}
\definecolor{myGold}{RGB}{231,141,20}
\definecolor{myBlue}{rgb}{0.19,0.41,.65}
\definecolor{myPurple}{RGB}{175,0,124}

\providecommand{\Comments}{0}
\ifnum\Comments>0
\usepackage[colorinlistoftodos,prependcaption,textsize=scriptsize]{todonotes}
\paperwidth=\dimexpr \paperwidth + 4.7cm\relax
\oddsidemargin=\dimexpr\oddsidemargin + 2.6cm\relax
\evensidemargin=\dimexpr\evensidemargin + 2.6cm\relax
\marginparwidth=\dimexpr\marginparwidth + 1.7cm\relax
\else
\usepackage[disable]{todonotes}
\fi

\newcommand{\mytodo}[1]{\ifnum\Comments=1{#1}\fi}

\newcommand{\tableoftodos}

\allowdisplaybreaks

\title{Differential Privacy Meets Fixed Parameter Tractability: Algorithms and Lower Bounds}
\date{\today}

\author{Pritish Kamath\\Google Research\\{\small \texttt{pritish@alum.mit.edu}}
\and
Ravi Kumar\\Google Research\\{\small \texttt{ravi.k53@gmail.com}}
\and
Pasin Manurangsi\\Google Research\\{\small \texttt{pasin@google.com}}
}

\begin{document}

\maketitle

\begin{abstract}
We study combinatorial optimization problems under the constraint of $\epsilon$-differential privacy ($\epsilon$-DP).  Given the strong lower bounds for explicitly outputting solutions, we work within the \emph{implicit representation} framework of Gupta et al. (SODA 2010), where a private polynomial-time randomized ``encoder'' generates a representation of a solution, and a ``decoder'' uses this representation along with the input to extract a valid final solution. 

In this work, we generalize this framework by allowing the encoder to run in fixed-parameter tractable time.  This circumvents  approximation barriers inherent to polynomial-time algorithms and 
obtains improved guarantees for many fundamental  combinatorial optimization problems.

Finally, we establish the first representation-independent lower bounds for our framework. Assuming a non-uniform variant of the Gap Exponential Time Hypothesis, 
for sufficiently small $\epsilon > 0$, we prove that no $\epsilon$-DP encoder--decoder pair can achieve certain approximation guarantees, if the decoder runs in subexponential time. We further provide representation-dependent lower bounds that hold even for larger $\epsilon$.
\end{abstract}

\section{Introduction}

Data privacy has become a critical requirement in the modern era of data analysis. As organizations collect vast amounts of sensitive information---ranging from social interactions to medical records---the risk of information leakage has motivated the adoption of Differential Privacy (DP) \cite{dwork2006calibrating}. DP provides a mathematically rigorous framework that ensures the output of an algorithm reveals little about any individual in the dataset, protecting against even adversaries with arbitrary auxiliary information. DP is formally defined as follows.

\begin{definition}[Differential Privacy \cite{dwork2006calibrating}]
\label{def:dp}
A randomized algorithm $\mathcal{M}$ is \emph{$\eps$-differentially private ($\eps$-DP)} if for all neighboring datasets $X \sim X'$ and all sets $\mathcal{S} \subseteq \mathrm{range}(\mathcal{M})$ of possible outcomes, we have: $\Pr[\mathcal{M}(X) \in \mathcal{S}] \le e^\eps \Pr[\mathcal{M}(X') \in \mathcal{S}]$,
where the probability is over the randomness of $\mathcal{M}$.
\end{definition}

A fundamental area of computer science with broad applications is combinatorial optimization. Many real-world problems can be naturally modeled as graph optimization tasks, such as finding influential individuals in a social network or clustering related entities. Since these graphs often encode sensitive relationships (e.g., communication patterns, financial transactions), solving combinatorial optimization problems with DP is an important research direction.

However, designing private algorithms for combinatorial optimization problems poses  significant challenges. For many problems, \emph{explicitly} outputting a solution (e.g., a specific subset of vertices) either inherently violates privacy or leads to trivial utility guarantees due to strong lower bounds. This difficulty is especially pronounced for problems with hard constraints (e.g., Vertex Cover), where every edge must be covered. Any algorithm that directly outputs a candidate subset with the uncertainty required by DP might fail to cover all edges, thereby producing an invalid solution.

To bypass these limitations,~\cite{GLMRT10} introduced the \emph{implicit representation} framework.  Here, the DP algorithm (``encoder'') does not output the solution directly. Instead, it privately publishes an \emph{implicit}  representation of the solution. A deterministic, non-private algorithm (``decoder'') then uses this representation, along with the original input, to construct the final valid solution. 

For the Vertex Cover problem,~\cite{GLMRT10} proposed a permutation-based representation, where the private encoder outputs a permutation of the vertices. Given this permutation, the decoder constructs a vertex cover by iterating through the edges and greedily picking the earliest endpoint vertex according to the permutation. They showed that this approach can achieve an expected approximation ratio of $2 + O(1/\eps)$. 
As $\eps \to \infty$, this  ratio approaches $2$, which is the best possible for any polynomial-time algorithm,  assuming the Unique Games Conjecture (UGC) \cite{KhotR08}.

Besides Vertex Cover, several other combinatorial optimization problems have also been studied in the model, including set cover, partial vertex cover, and facility location~\cite{GLMRT10,EGLW19,CEFG22,LNV23,GhaziKKMS24,Manurangsi25}. As~\cite{CEFG22} noted, the Joint DP model~\cite{KearnsPRU14} is also a form of implicit representation; thus, Joint DP results on matching, packing, allocation (e.g.,~\cite{HsuHRRW16,Hsu0RW16,HuangZ18}) can be 
interpreted in the implicit representation framework.

\subsection{Our Contributions}

We generalize the implicit representation framework and obtain new upper and lower bounds for many private combinatorial optimization problems. In particular, we generalize as follows:
\begin{itemize}[nosep]
\item \textbf{General Representation and Decoder:} The encoder can output an \emph{arbitrary}  implicit representation and the decoder can be any \emph{polynomial-time} algorithm; only specific representations (e.g., permutations or assignments) and decoders have been studied so far.
\item \textbf{Encoder:} The encoder can run in \emph{Fixed-Parameter Tractable (FPT)} time; previous work mainly focused on polynomial-time encoders.
\end{itemize}
Under these generalizations, we obtain the following results.
\begin{itemize}[nosep]
  \item \textbf{Upper Bounds via FPT:} 
  We show that the barriers inherent to polynomial-time encoders can be overcome by our generalized FPT encoder.  For example, for Vertex Cover, we show how to achieve $\left(1 + \tO\left(1/\eps\right)\right)$-approximation in FPT time. For sufficiently large $\eps > 0$, this overcomes the UGC-hardness of approximation for Vertex Cover, as discussed next.
    \item \textbf{Lower Bounds via Gap-ETH:} Despite our generalized decoder, we show a strong lower bound: Under a non-uniform version of Gap-ETH and for sufficiently small $\eps, \gamma > 0$, any $\eps$-DP algorithm in this model cannot achieve $(1 \pm \gamma)$-approximation for several problems, including Vertex Cover.
\end{itemize}

Below we elaborate on these contributions, starting with the formal definition of the implicit representation model.

\subsubsection{Implicit Representation Model}

We formalize a generalized model of privacy for combinatorial optimization. Rather than outputting an explicit solution directly---which may be difficult to do privately with high utility---the DP algorithm outputs an \emph{implicit representation} of a solution. This representation is then decoded into a valid solution using the public knowledge and the original input. 

\begin{definition}[Implicit Representation Scheme]
\label{def:implicit-rep}
An $\eps$-DP \emph{implicit representation scheme} for a problem $\Pi$ consists of two algorithms:
\begin{enumerate}[nosep]
    \item \textbf{Encoder $\mathcal{E}$}: A randomized $\eps$-DP algorithm that takes an input dataset $X$ and outputs an object $C$ from a representation space $\mathcal{C}$.
    \item \textbf{Decoder $\mathcal{D}$}: An algorithm that takes the input dataset $X$ and the object $C \in \mathcal{C}$, and outputs a valid solution for the instance $X$ of problem $\Pi$.
\end{enumerate}
The overall algorithm $\mathcal{A}$ for the problem $\Pi$ is defined by the composition $\mathcal{A}(X) = \mathcal{D}(X, \mathcal{E}(X))$.
\end{definition}

Because the decoder $\mathcal{D}$ has access to the full input dataset $X$, the final solution $\mathcal{D}(X, \mathcal{E}(X))$ need \emph{not} be DP with respect to $X$. Indeed, this is a crucial aspect of this model; as demonstrated in \cite{GLMRT10} requiring the final solution to be DP rules out any non-trivial solution for many covering and packing tasks with hard constraints. 

\paragraph{Efficient Decoder Assumption.} Without any further restriction, coming up with an $\eps$-DP implicit representation scheme is trivial: Let the encoder $\cE$ do nothing, and let $\cD$ solve the problem directly (e.g., via brute-force). Indeed, in all previous models, $\cD$ are often ``simple'' algorithms. For the purpose of this work, we assume throughout that $\cD$ runs in polynomial time.

\begin{assumption} \label{as:decoder-polytime}
The decoder $\cD$ runs in polynomial time, with respect to the size of the input $X$.
\end{assumption}

Unless stated otherwise, we work under this assumption throughout this paper. 
As we remark in \Cref{sec:conclusion}, it is possible to make other (e.g., more restrictive) assumptions on $\cD$. However, we find polynomial-time to give a clean abstraction, which allows us to both achieve meaningful upper and lower bounds.

\paragraph{Adjacency Notion.}
We focus on (hyper)graph algorithms in our paper, and we use \emph{(hyper)edge adjacency} notion throughout this work. Namely, the input dataset $X$ is always a hypergraph $G = (V, E)$. Two hypergraphs $G = (V, E)$ and $G' = (V, E')$ are \emph{neighbors}, denoted by $G \sim G'$, if they differ in the presence or absence of a single (hyper)edge, i.e., $|E \triangle E'| \le 1$.

We emphasize, however, that the general definition of the implicit representation model (Definition~\ref{def:implicit-rep}) is agnostic to the choice of neighborhood relation and can naturally be applied with other adjacency notions (such as node DP) and non-graph problems as well.

\subsubsection{Upper Bounds}

Working under this model, we give several new approximation algorithms for graph combinatorial optimization problems. At a high-level, our algorithms (specifically, the encoders) use FPT time to overcome inapproximability barriers.

Recall that an algorithm is said to be FPT\footnote{For more detail on parameterized algorithms and complexity, please refer to textbooks on the topic, e.g., \cite{DowneyF13}.} with respect to parameter $k$ if it runs in time $f(k) \cdot N^{O(1)}$ where $f$ can be any function and $N$ denotes the instance size. Throughout our work, we let $k$ be the so-called \emph{natural} parameter, i.e.,  the desired optimum of the solution.

All problems and our algorithms below share the following characterizations:
\begin{itemize}[nosep]
\item The problems are APX-hard, i.e.,  unless NP $\subseteq$ BPP, for some $\gamma > 0$, there is no polynomial-time randomized algorithm that achieves $(1 \pm \gamma)$-approximation.
\item For any $\gamma > 0$, our algorithms achieve $(1 \pm \gamma)$-approximation in FPT time for sufficiently large $\eps > 0$.
In other words, FPT time allows us to circumvent the aforementioned APX-hardness.
\end{itemize}

\paragraph{Vertex Cover and $d$-Hitting Set.}
We start by considering the classic Vertex Cover problem. Utilizing the same permutation-based implicit representation model as~\cite{GLMRT10}, we break the aforementioned factor-$2$ approximation barrier by allowing the encoder to run in FPT time. Specifically, we achieve a $(1 + \tO(1/\eps))$-approximation algorithm where the encoder runs in FPT time parameterized by the size of the optimal vertex cover. 

At a high level, this is achieved by first observing that the algorithm of \cite{GLMRT10} actually yields an optimal solution, but with an exponentially small probability. We boost this success probability by repeatedly running the algorithm exponentially many times to generate a large pool of candidate representations. We then use private hyperparameter tuning (i.e., private selection from private candidates)~\cite{LT19,Papernot022,GhaziKKKMZ25} to privately identify a representation that yields a near-optimal vertex cover. Furthermore, we observe that this algorithm naturally extends to the $d$-Hitting Set problem, which can be viewed as the Vertex Cover problem on $d$-uniform hypergraphs.


\paragraph{Color Coding for $H$-Packing.}
Next, we show how to adapt the classic FPT technique of Color Coding~\cite{AlonYZ94} into the DP setting. We demonstrate this on the subgraph packing problem: Given a fixed graph $H$, find a maximum vertex-disjoint copies of $H$ in $G$. A standard approach to solve this non-privately is to randomly color the vertices into $k$ colors. Then, try to find a copy of $H$ within each color. Each random color will succeed with probability $k^{-O(k)}$.

We adapt this algorithm for our setting by picking $k^{O(k)}$ many such colors and use the exponential  mechanism to select a color that yields a near-optimal packing. This yields an $\Paren{1 - \tO\Paren{\frac{\log k}{\eps}}}$-approximation for the problem.

\paragraph{Subset Representation for Induced Subgraphs.}
Finally, we consider the problem of finding maximum induced subgraphs satisfying hereditary properties on sparse graphs. This includes the problem of finding a maximum independent set. Due to the sparsity of the graph, the optimum is always large, i.e., $k \geq \Omega(|V|)$. Thus, the brute-force $2^{O(|V|)}$-time algorithm is an FPT solution in the non-private setting.

To achieve privacy, we introduce the ``subset'' representation. Here, the encoder uses the exponential mechanism~\cite{MT07} to privately outputs a small subset of vertices, and the decoder efficiently extracts a large induced subgraph from this set in polynomial time. This yields $(1 - O(1/\eps))$-approximation for the problem.

\subsubsection{Lower Bounds}

It is natural to ask whether one can improve the privacy-utility tradeoff of our algorithms. We make progress in this direction by proving both representation-independent and representation-dependent lower bounds.

\paragraph{Representation-Dependent Lower Bounds.}
First, we consider the representation-dependent setup, where the implicit representation format and the decoder algorithm are fixed to be the same as in the algorithms. In this setting, we show a tradeoff between the privacy parameter and the approximation ratio for all problems studied here. Similar to previous work~\cite{GLMRT10}, this is proved via simple packing arguments~\cite{HardtT10}.

\paragraph{Representation-Independent Lower Bounds.}
Given the lower bound above, the next question is whether it is possible to improve our algorithm further by inventing new implicit representations that is more amenable to privacy while still allows efficient decoding. 

We provide some initial negative results:
Assuming a non-uniform variant of the Gap Exponential Time Hypothesis~\cite{Dinur16,ManurangsiR17} (Gap-ETH with advice), we rule out any $\eps$-DP implicit representation scheme from achieving $(1 \pm \gamma)$ approximations, when $\eps, \gamma > 0$ are sufficiently small. 

The main challenge in proving such results is that the encoder's running time is not restricted (e.g., could be super-exponential). Thus, we cannot simply try to run it to contradict Gap-ETH directly. Instead, the core of our representation-independent proof relies on privacy of the representation to shows that there exists a small set (i.e., ``covering'') of representations such that a decoder must be successful on at least one of them on \emph{any input}. This ``covering'' can serve the non-uniform advice string, and we (approximately) solve 3SAT by running only the decoder (bypassing the encoder completely), ultimately contradicting Gap-ETH with advice.

The identification of covering above is highly reminiscent of lower bound proofs against instance compression (aka \emph{kernels})~\cite{FortnowS11,BodlaenderDFH09}. However, our proof is unique as the covering of representations exists due to the \emph{privacy} of the representations. To the best of our knowledge, this is a novel proof technique unlike any previous computational lower bound proofs in DP---which either rely on cryptographic assumptions (e.g., \cite{dwork2009complexity,UllmanV20,GhaziIK0M23,GhaziGKKKM26}) or simply re-run the entire algorithm multiple times until a desired solution is found (e.g., \cite{GeorgievH22,HillebrandMSV26-voting}).


\section{Preliminaries}
\label{sec:preliminaries}

For any positive integer $q$, Let $[q]$ denote $\{1, \dots, q\}$. We use $\ind{W}$ to denote the indicator variable of an event $W$. For any (hyper)graph $G = (V, E)$, we use $n$ and $m$ to denote the number of vertices and the number of (hyper)edges, respectively. The induced subgraph $G[S]$ for $S \subseteq V$ is the graph $(S, E[S])$ where $E[S] := \{e \in E \mid e \subseteq S\}$. The \emph{degree} of a vertex $v \in V$ with respect to a set $E$ is defined as $\deg_E(v) := |\{e \in E \mid v \in e\}|$.

We also recall the following probability distributions, which will be useful in our algorithms:
\begin{itemize}[nosep]
\item For $p \in [0, 1]$, the \emph{Bernoulli distribution}  $\Ber(p)$ is supported on $\{0, 1\}$ where its probability mass function at 1 is equal to $p$.
\item For $p \in [0, 1)$, the \emph{Geometric distribution}  $\Geo(p)$ is supported on $\Z_{\ge 0}$  and its probability mass function at $i$ is equal to $p^i(1 - p)$.
\item For $b > 0$, the \emph{Laplace distribution}  $\Lap(b)$ is supported on $\R$ and its probability density function at $x$ is $\frac{1}{2b} \cdot e^{-|x|/b}$.
\end{itemize}

\subsection{Differential Privacy} 

Recall also the definitions of DP (\Cref{def:dp}) and Implicit Representation Scheme (\Cref{def:implicit-rep}). For our representation-independent lower bounds (\Cref{sec:lower-bounds}), we assume that the representation $C$ is encoded as a bit string and that the decoder running time is at least the length of this bit string.

Below, we give some additional relevant definitions and background from the DP literature.

We use $\cX$ to denote the set of all input datasets.

\begin{definition}[Sensitivity] 
The \emph{sensitivity} of a function $f: \cX \to \R$ is defined as $\Delta(f) := \max_{X \sim X'} |f(X) - f(X')|$ where the maximum is over all neighboring datasets $X, X'$.
\end{definition}

A simple way to achieve DP is by adding Laplace noise calibrated to the sensitivity of the function:

\begin{definition}[Laplace Mechanism~\cite{dwork2006calibrating}] \label{def:laplace_mech}
The \emph{Laplace mechanism} on input $X$ for function $f: \cX \to \R$ outputs $f(X) + \eta$ where $\eta \sim \Lap(\Delta(f)/\eps)$. This mechanism is $\eps$-DP.
\end{definition}

We will also use the exponential mechanism \cite{MT07}, which we recall below.

\begin{definition}[Exponential Mechanism~\cite{MT07}] \label{def:em}
Given a candidate set $\cC$ and, for every $c \in \cC$, a scoring function $\scr(\cdot, c): \cX \to \R$. The \emph{exponential mechanism} $\EM_{\scr}(X)$ on input $X$ samples from $\cC$ where each $c \in \cC$ is picked  with probability $\propto \exp\Paren{\frac{\eps}{2}\cdot\scr(X, c)}$.
\end{definition}

\begin{theorem}[\cite{MT07}] \label{thm:em}
Suppose $\Delta(\scr(\cdot; c)) \leq 1$. Then, the $\EM_{\scr}$ is $\eps$-DP. 
Furthermore, if $\tc$ is the output from $\EM_{\scr}(X)$, then it satisfies $\E[\scr(X, \tc)] \geq \max_{c \in \cC} \scr(X, c) - O\Paren{\frac{\log |\cC|}{\eps}}$.
\end{theorem}

\subsection{(Parameterized) Approximation Algorithms}
We evaluate the performance of an implicit representation scheme based on the expected objective value of the final decoded solution. Let $\opt_{\Pi}(X)$ denote the optimal objective value of the problem $\Pi$ on instance $X$, and let $\val_{\Pi}(X, S)$ denote the objective value of a valid solution $S$ on $X$. When $\Pi$ is clear from context, we drop the subscript $\Pi$ for readability. Throughout this work, we assume that $\val_{\Pi}(X, S)$ is a non-negative integer bounded by $N^{O(1)}$ where $N$ is the size of the input $X$.

Throughout this work, we consider the parameterized algorithms where the parameter is always the optimum. Here we allow the encoder $\cE$ to take in\footnote{We allow the encoder to take $k$ as input for convenience. In all problems studied in this work, the optimum has low sensitivity. Thus, it can be computed (with low error) using, e.g., the Laplace mechanism (\Cref{def:laplace_mech}).} the parameter $k$, in addition to the input dataset $X$. We only require the solution when the given parameter $k$ is an under-estimation (resp., over-estimation) for the optimum in maximization (resp., minimization) problems. This can be formalized as follows\footnote{For more background on parameterized approximation algorithms and hardness, see, e.g.,  \cite{FeldmannSLM20}.}. The expectation is taken over the internal randomness of the encoder $\mathcal{E}$.

\begin{definition}[Parameterized Approximation Algorithm]
\label{def:approx-max}
For a maximization (resp.,  minimization) problem $\Pi$, an implicit representation scheme $(\cE, \cD)$ is an $\alpha$-approximation algorithm for $\alpha \leq 1$ (resp., $\alpha \geq 1$) if the following holds: For any $k \leq \opt(X)$ (resp.,  $k \geq \opt(X)$), we have $\E_{S \sim \cD(X, \cE(X; k))}[\val(X, S)] \geq \alpha \cdot k$ (resp., $\leq \alpha \cdot k$). 

If additionally $\cE(X; k)$ runs in $f(k) \cdot N^{O(1)}$ time, where $f$ is some function and $N$ is the size of the input $X$, then we say that $(\cE, \cD)$ is an $\alpha$-FPT-approximation algorithm.
\end{definition}
Recall that, unless stated otherwise, we assume that the decoder $\cD$ runs in polynomial time (\Cref{as:decoder-polytime}). 

\paragraph{Gap Problems.} For our lower bound proofs, it will be more convenient to work with ``gap problems'', which can be defined as follows.
\begin{definition}[Gap Problem]
For a maximization (resp. minimization) problem $\Pi$, the \emph{$\rho$-gap-$\Pi$} is to, given $(X, k)$ distinguish between the following two cases:
\begin{itemize}[nosep]
\item (Yes) $\opt(X) = k$
\item (No) $\opt(X) < \rho \cdot k$ (resp. $\opt(X) > \rho \cdot k$)
\end{itemize}
We say that a randomized algorithm $\alg$ solves the gap problem with one-sided error  
if the following holds for some constant $\xi > 0$:
\begin{itemize}[nosep]
\item If $\opt(X) = k$, $\Pr[\alg(X) = YES] \geq \xi$.
\item If $\opt(X) < \rho \cdot k$ (resp. $\opt(X) > \rho \cdot k$), $\Pr[\alg(X) = NO] = 1$.
\end{itemize}
We say that a deterministic algorithm solves the gap problem when similar guarantees hold (where there is now no error).
\end{definition}

It is well known that an approximation algorithm can be easily turned to an algorithm for the gap problem. Below we formulate this statement for implicit representation schemes. While the proof is standard, we include it here for completeness.

\begin{observation} \label{obs:exp-to-prob-approx}
Let $\Pi$ be a maximization (resp. minimization) problem, and let $\alpha, \rho > 0$ be constants such that $\alpha > \rho$ (resp. $\alpha < \rho$). If there is an $\eps$-DP implicit representation scheme with a subexponential-time decoder that achieves $\alpha$-approximation for $\Pi$, then there is an $\eps$-DP implicit representation scheme with a subexponential-time decoder that solves the $\rho$-gap-$\Pi$ problem with one-sided error. 
\end{observation}

\begin{proof}
Let us start by considering the maximization problem $\Pi$. Suppose that there is a $\eps$-DP $\alpha$-approximation algorithm $(\cE, \cD)$ for $\Pi$. Our algorithm $(\cE, \cD')$ for $\rho$-gap-$\Pi$ keeps the encoder $\cE$ unchanged. The decoder simply is now as follows: First, run the original decoder $\cD$ to obtain a solution $S$. It then checks whether $\val(X, S) \geq \rho \cdot k$, return YES. Otherwise, return NO.

Note that the running time of the decoder increases by at most $n^{O(1)}$.

To see that this solves the gap problem, consider the two cases:
\begin{itemize}[nosep]
\item If $\opt(X) = k$, from the approximation guarantee of $(\cE, \cD)$, we have $\E[\val(X, S)] \geq \alpha \cdot k$. Meanwhile, $\val(X, S)$ is a number between $[0, k]$. Thus, by applying Markov's inequality to $k - \val(X, S)$, we have $\Pr[\val(X, S) \geq \rho \cdot k] \geq \alpha - \rho$. In other words, $(\cE, \cD')$ returns YES with probability at least $\alpha - \rho$.
\item If $\opt(X) < \rho \cdot k$, then by definition there is no solution $S$ with $\val(X, S) \geq \rho \cdot k$. Thus, the algorithm always return NO.
\end{itemize}

The minimization case is similar: $\cD'$ returns YES iff $\val(X, S) \leq \rho \cdot k$. If $\opt(X) = k$, by Markov's inequality, the algorithm returns YES with probability at least $1 - \alpha / \rho$. Meanwhile, in the case $\opt(X) > \rho \cdot k$, we always return NO.
\end{proof}

\subsection{Problem Definitions}
We study several classical combinatorial optimization problems on hypergraphs in this framework. As mentioned in the introduction, the input to all these problems are (hyper)graphs and we consider two (hyper)graphs to be \emph{neighbors} (for purpose of DP) if they differ in a single (hyper)edge. Below, we formally define the objectives for all problems we will study.

\paragraph{$d$-Hitting Set (Hypergraph Vertex Cover)}: Given a hypergraph $G = (V, E)$ where each $e \in E$ has size at most $d$, find a minimum-size subset $S \subseteq V$ such that $S \cap e \neq \emptyset$ for all $e \in E$.
\begin{itemize}[nosep]
\item \textbf{Vertex Cover} is the special case of $d$-Hitting Set where $d = 2$.
\end{itemize}

The Vertex Cover problem is among the first problems shown to be in FPT \cite{BussG93} and is often used as an illustrative example (e.g.,  \cite{DowneyF13,CyganFKLMPPS15}). As alluded to earlier, it is also well studied from the hardness of approximation standpoint; the problem is known to be UGC-hard to approximate to within a factor of $2 - \eps$~\cite{KhotR08} and NP-hard to approximate to within a factor of $\sqrt{2} - \eps$~\cite{KhotMS17,DinurKKMS18a}. For the hypergraph version, the known hardness factors become $d - \eps$ and $d - 1 - \eps$, respectively~\cite{KhotR08,DinurGKR05}.

\paragraph{Maximum Induced Subgraph (with Hereditary Property):} Given a graph $G=(V,E)$ and a hereditary graph property\footnote{A graph property is said to be \emph{hereditary} if it is closed under vertex deletion.} $\Gamma$, find a maximum-size $S \subseteq V$ such that $G[S]$ satisfies $\Gamma$.
\begin{itemize}[nosep]
\item \textbf{Independent Set} is a special case of Maximum Induced Subgraph (with Hereditary Property) where $\Gamma$ is the property of being an independent set.
\end{itemize}

For this problem, our algorithm requires that the graph $G$ is \emph{$\Delta$-degenerate} (i.e., for every non-empty $S \subseteq V$, $G[S]$ has at least one vertex with degree at most $\Delta$) for some constant $\Delta \in \N$. Such a sparsity assumption is necessary due to strong FPT-inapproximability results for the problem in the general case \cite{ChalermsookCKLM20}, which rules out any non-trivial $o(k)$-FPT-approximation algorithm.

We note that the problem remains NP-hard in this setting. Specifically, Independent Set is APX-hard even on degree-3 graphs~\cite{AlimontiK00}.

\paragraph{$H$-Packing:} Given a hypergraph $G=(V, E)$ and a fixed hypergraph $H$, find a maximum-size collection of vertex-disjoint sub(hyper)graphs in $G$, each isomorphic to $H$.

\paragraph{Triangle-Transversal:} Given an undirected graph $G=(V,E)$, find a minimum-size subset $S \subseteq V$ to delete such that $G[V \setminus S]$ contains no triangles.

Again, $H$-Packing/Transversal is well studied in the approximation algorithms  literature, and many algorithms and hardness results are known; see \cite{GuruswamiL17} and references therein. In particular, both Triangle-Packing and Triangle-Transversal problems are known to be APX-hard even on degree-4 graphs~\cite{Kann91,RooijNB13,GuptaLLMW19,GroshausHKNP11}.

Both problems are also known to be in FPT, and algorithms  for specific $H$ such as triangles, or paths are known (e.g.,~\cite{GrammGHN04,FernauR09,BjorklundHKK17}).



We note that for the Triangle-Transversal problem, we will only prove representation-dependent lower bounds for the problem. For the other problems, we will give algorithms together with representation-dependent and representation-independent lower bounds. 

\subsection{Gap-ETH with Advice}

For our representation-independent lower bounds, we will use a non-uniform variant of the Gap Exponential Time Hypothesis (Gap-ETH) \cite{Dinur16, ManurangsiR17}. To formalize this, we start by defining the classical Gap-3SAT problem.

\begin{definition}[Bounded Occurrence Gap-3SAT] \label{def:gap-3sat}
An instance of 3SAT($D$) consists of a 3-CNF formula $\phi$ on $N$ variables where each variable appears in at most $D$ clauses. For a constant $\gamma \in (0, 1)$, the $\gamma$-Gap-3SAT($D$) problem asks, given such an instance $\phi$, to distinguish between
\begin{itemize}[nosep]
\item (Yes) There exists an assignment that satisfies all clauses in $\phi$.
\item (No) Every assignment satisfies at most $\gamma$ fraction of clauses in $\phi$.
\end{itemize}
\end{definition}

The Gap-ETH with Advice is a strengthening of the standard Gap-ETH~\cite{Dinur16,ManurangsiR17}. To the best of our knowledge, this assumption has not been studied before in literature. However, given that advice strings usually do not (significantly) help with solving NP-hard problems, we consider Gap-ETH with Advice to be nearly as plausible as the standard Gap-ETH.

\begin{assumption}[Gap-ETH with Advice]
\label{assump:gap-eth-advice}
There exist constants $D \ge 3$, $\gamma \in (0, 1)$, and $c_0 > 0$ such that no $O(2^{c_0 N})$-time algorithm, even when provided with a non-uniform advice string\footnote{As is standard in the area, the advice string can only depend on $N$ and nothing else.} of size $O(2^{c_0 N})$, can solve $\gamma$-Gap-3SAT($D$) on $N$ variables.
\end{assumption}

\section{Algorithms}
\label{sec:alg}
\subsection{Vertex Cover and $d$-Hitting Set}
\label{sec:vc}
In this section, we present an algorithm for the $d$-Hitting Set problem that achieves the following: 

\begin{theorem}
\label{thm:vc-main}
There is an $\eps$-DP $(1 + O(\frac{\log(d + 1/\eps)}{\eps}))$-FPT-approximation implicit representation scheme for $d$-Hitting Set with running time $(d + 1/\eps)^{O(k)} \cdot n^{O(1)}$.
\end{theorem}

In fact, our implicit representation is the same as that of~\cite{GLMRT10}.  The approximation ratio can be viewed as a direct improvement over their work, albeit at the cost of FPT running time (instead of polynomial time). The implicit representation we use is defined below.

\begin{definition}[Permutation Representation] \label{def:perm-rep}
In the \emph{permutation representation}, the encoder outputs a permutation $\pi$ of the vertex set $V$ as the implicit representation. The decoder $\cD(G, \pi)$ operates by selecting the earliest element in $\pi$ for each hyperedge $e \in E$, forming the hitting set $S_\pi(E) \coloneqq \{\min_\pi(e) \mid e \in E\}$. 
\end{definition}

The algorithm extends the \textsc{DP-FPT-VC} generator of~\cite{GLMRT10} to hypergraphs and uses the Generalized AboveThreshold mechanism \cite{GhaziKKKMZ25} to select a candidate.

\subsubsection{A Low Probability Algorithm}

As a first step of our algorithm, we show that the noisy degree sampling algorithm of~\cite{GLMRT10} natively outputs the optimal hitting set, albeit with an exponentially small probability. 

Their algorithm works as follows: at each iteration $i = 1, \dots, n$, we sample a remaining vertex based on its current (noised) degree using a noise parameter $w_i$. This vertex is appended to the permutation; any hit hyperedges are removed. The full algorithm is presented in \Cref{alg:dp-fpt-dhs-gen}.

\begin{algorithm}[H]
\caption{Noisy Degree Sampling (\textsc{DP-FPT-dHS-Gen})~\cite{GLMRT10}}
\label{alg:dp-fpt-dhs-gen}
\begin{algorithmic}[1]
\REQUIRE Hypergraph $G = (V, E)$ where each edge has size at most $d$, privacy parameter $\eps_0 > 0$.
\STATE Let $V_1 \gets V$ and $E_1 \gets E$.
\FOR{$i = 1, \dots, n$}
\STATE Compute the  noise term:
$w_i \gets \frac{4d}{\eps_0} \sqrt{\frac{n}{n-i+1}}$
\STATE Pick a vertex $c^*_i \in V_i$ with probability proportional to its degree in the remaining hypergraph plus the noise term $w_i$:
\[
\Pr[c^*_i = v] = \frac{\deg_{E_i}(v) + w_i}{\sum_{u \in V_i} \deg_{E_i}(u) + (n-i+1)w_i}.
\]
\STATE Update $V_{i+1} \gets V_i \setminus \{c^*_i\}$ and $E_{i+1} \gets \{e \in E_i \mid c^*_i \notin e\}$.
\ENDFOR
\RETURN permutation $\pi = (c^*_1, \dots, c^*_n)$.
\end{algorithmic}
\end{algorithm}

Let $S^*$ be any fixed optimal hitting set of size $k$. 
If there were no noise (i.e., $w_i = 0$), then it is relatively simple to see that the vertices in $S^*$ will be picked as the first $k$ vertices with probability at least $d^{-k}$. When this happens, we of course have that $S_{\pi}(E) = S^*$. Below, we conduct a refined utility analysis to show that even with the increasing noise added, we still have $S_{\pi}(E) = S^*$ with a pure parameterized success probability.

\begin{lemma}
\label{lem:dhs-gen}
\textup{\textsc{DP-FPT-dHS-Gen}} is $\eps_0$-DP. Furthermore, for any optimal hitting set $S^*$ of size $k$, it outputs a permutation $\pi$ such that $S_\pi(E) = S^*$ with probability at least $(d + 1/\eps_0)^{-O(k)}$.
\end{lemma}
\begin{proof}
We assume w.l.o.g. that $n \geq 3k$ as otherwise a random permutation suffices\footnote{The probability that all elements of $S^*$ come first in the permutation is $\frac{1}{\binom{n}{k}}$, which is $2^{-O(k)}$ for $n = O(k)$.}.

\textbf{Privacy:} The privacy argument
mirrors that of \cite{GLMRT10} for $d = 2$.

Consider two neighboring hyperedge sets $E, E'$ differing by one hyperedge $e$ of size at most $d$. Let $(x_1, \dots, x_n)$ be any possible output permutation. Let $i^* \in [n]$ be the first step where a vertex $x_{i^*} \in e$ is selected. For $i < i^*$, the hyperedge sets $E_i$ and $E'_i$ differ exactly by $e$, and the degrees of vertices outside $e$ are identical. The ratio of probabilities of selecting $x \notin e$ is:
\begin{align*}
\frac{\Pr[c^*_i = x_i \mid E_i]}{\Pr[c^*_i = x_i \mid E'_i]} &= \frac{\deg_{E_i}(x) + w_i}{\deg_{E'_i}(x) + w_i} \cdot \frac{\sum_{u \in V_i} \deg_{E'_i}(u) + (n-i+1)w_i}{\sum_{u \in V_i} \deg_{E_i}(u) + (n-i+1)w_i} \\ &\le 1 \cdot \left(1 + \frac{d}{\sum_{u \in V_i} \deg_{E_i}(u) + (n-i+1)w_i}\right) \le \exp\left( \frac{d}{(n-i+1)w_i} \right).
\end{align*}
where, in the first inequality, we use the fact that $e$ has size at most $d$, meaning that the total degree on the two hypergraphs differ by at most $d$.

At step $i^*$, the ratio for selecting $x_{i^*} \in e$ is bounded by:
\begin{align*}
\frac{\Pr[c^*_{i^*} = x_{i^*} \mid E_{i^*}]}{\Pr[c^*_{i^*} = x_{i^*} \mid E'_{i^*}]} &=
\frac{\deg_{E_{i^*}}(x) + w_{i^*}}{\deg_{E'_{i^*}}(x) + w_{i^*}} \cdot  \frac{\sum_{u \in V_{i^*}} \deg_{E'_{i^*}}(u) + (n-{i^*}+1)w_{i^*}}{\sum_{u \in V_{i^*}} \deg_{E_{i^*}}(u) + (n-{i^*}+1)w_{i^*}} \\ &\le \left(1 + \frac{1}{w_{i^*}}\right) \exp\left( \frac{d}{(n-i^*+1)w_{i^*}} \right) \\ &\le \exp\left( \frac{1}{w_{i^*}} + \frac{d}{(n-i^*+1)w_{i^*}} \right).
\end{align*}
For all $i > i^*$, the remaining hypergraphs are identical ($E_i = E'_i$), so the ratio is exactly $1$. 

Thus, we have
\begin{align*}
\frac{\Pr[\pi = (x_1, \dots, x_{n}) \mid E]}{\Pr[\pi = (x_1, \dots, x_{n}) \mid E']} &\leq \exp\left(\frac{1}{w_{i^*}} + \sum_{i=1}^{i^*} \frac{d}{(n - i + 1)w_i}\right) \\
&\leq \exp\Paren{\frac{\eps_0}{4d} + \frac{\eps_0}{4} \cdot \sum_{i \in [n]} \sqrt{\frac{1}{n(n - i + 1)}}} \leq \exp(\eps_0).
\end{align*}
Thus, the permutation $\pi$ is $\eps_0$-DP.

\textbf{Utility:}
Let $S^*$ be a fixed optimal hitting set of size $k$. At any step $i \in [n]$, 
we partition $V_i$ into three sets:
\begin{itemize}[nosep]
    \item \emph{Good} vertices $S^*_i = S^* \cap V_i$, of size $k_i = |S^*_i| \le k$.
    \item \emph{Bad} vertices $F_i = \{u \in V_i \setminus S^*_i \mid \deg_{E_i}(u) > 0\}$, of size $f_i = |F_i|$.
    \item \emph{Irrelevant} vertices $I_i = \{u \in V_i \setminus S^*_i \mid \deg_{E_i}(u) = 0\}$, of size $n_i = |I_i|$.
\end{itemize} 
Observe that any irrelevant vertices 
will never be chosen by the decoder $\cD(G, \pi)$. 
Therefore, the insertion of irrelevant vertices into $\pi$ has no impact on the size of the decoded hitting set $S_\pi(E)$. Thus, to guarantee that $S_\pi(E) = S^*$, it suffices to establish the following success event $\mathcal{E}^*$: \emph{all $k$ vertices of $S^*$ are selected before any bad vertex is selected.}
We prove a lower bound on this event in the claim below.

\begin{claim} \label{claim:good-event}
Let $q := \frac{1}{d(1 + 6d/\eps_0)}$. Then, $
\Pr[\cE^*] \geq q^k \cdot \frac{\binom{\lfloor n/2 \rfloor}{k}}{\binom{n}{k}}$.
\end{claim}

\begin{proof}[Proof of \Cref{claim:good-event}]
Let $M := \lfloor n/2 \rfloor$ and $N_0 := n - M = \lceil n/2 \rceil$. For any $i \in [M+1], \ell \in \{0, \dots, k\}$, we say that a history $(c^*_1, \dots, c^*_{i-1})$ is \emph{$(i,\ell)$-successful} if no bad vertex has been selected so far and $\ell$ good vertices remained. 
Furthermore, we let $\cE^*_{\leq M}$ denote the event where all good vertices are selected by step $M$ before any bad vertex is selected.

We will prove an even stronger statement:
Conditioned on any $(i,\ell)$-successful history $\cH_{i-1}$, the probability of $\cE^*_{\leq M}$ satisfies\footnote{Here we use the convention $\binom{0}{\ell} = \ind{\ell = 0}$.}
\begin{equation} \label{eq:single-stage-ind}
    \Pr\left[\cE^*_{\leq M} \mid \cH_{i - 1} \right] \ge f(i, \ell) := q^\ell \cdot \frac{\binom{n - i + 1 - N_0}{\ell}}{\binom{n - i + 1}{\ell}}.
\end{equation} 
Note that the bound on $\Pr[\cE^*] \geq \Pr[\cE^*_M]$ simply follows by plugging in $i = 1, 
\ell = k$. 

We prove \eqref{eq:single-stage-ind} using a backward induction on the step $i = M+1, \dots, 1$.

\paragraph{Base Case:} For $i = M+1$, we simply have $f(i, \ell) = \ind{\ell = 0}$ and \eqref{eq:single-stage-ind} clearly holds.

\paragraph{Inductive Step:} Next, suppose \eqref{eq:single-stage-ind} holds for some $i + 1 \in \{2, \dots, M + 1\}$. We will show that this holds for $i$. First, note that, if $\ell = 0$ (i.e. all good vertices have already been selected), then $\Pr\left[\cE^*_{\leq M} \mid \cH_{i - 1} \right] = 1$ and the bound obviously holds. We henceforth consider the case $\ell \geq 1$. Consider any $(i,\ell)$-successful history $\cH_{i-1}$ and consider the decision at step $i$. 

At step $i$, let $r_i := \Pr[c_i^* \in S_i^* \cup F_i \mid \mathcal{H}_{i-1}]$ be the probability that step $i$ selects a relevant (good or bad) vertex, and let $g_i := \Pr[c_i^* \in S_i^* \mid \mathcal{H}_{i-1}, \, c_i^* \in S_i^* \cup F_i]$ be the conditional probability that a relevant selection is good.
At step $i$:
\begin{itemize}[nosep]
    \item With probability $r_i g_i$, a good vertex ($c_i^* \in S_i^*$) is selected, leaving $\ell - 1$ good vertices.
    \item With probability $1 - r_i$, an irrelevant vertex ($c_i^* \in I_i$) is selected, leaving $\ell$ good vertices.
    \item With probability $r_i(1 - g_i)$, a bad vertex ($c_i^* \in F_i$) is selected, failing the event $\cE^*_{\leq M}$.
\end{itemize}
By the inductive hypothesis, we have
\begin{equation} \label{eq:ind-step}
    \Pr\left[\cE^*_{\leq M} \mid \cH_{i - 1} \right] \ge r_i g_i \cdot f(i+1, \ell - 1) + (1 - r_i) \cdot f(i+1, \ell).
\end{equation}
Since $i \leq M$, we have $w_i \leq 6d/\eps_0$.
Moreover, since $S_i^*$ is a hitting set of $E_i$, we have $\sum_{v \in S_i^*} \deg_{E_i}(v) \ge |E_i|$. Meanwhile, since each edge in $E_i$ has size at most $d$, we have\footnote{Note that since $S^*$ is an optimal hitting set and no bad vertex has been selected, every remaining good vertex has degree at least one in $E_i$.} $\ell + f_i \le \sum_{u \in S_i^* \cup F_i} \deg_{E_i}(u) \le d|E_i|$. Thus,
\[
g_i = \frac{\sum_{v \in S_i^*} (\deg_{E_i}(v) + w_i)}{\sum_{u \in S_i^* \cup F_i} (\deg_{E_i}(u) + w_i)} \ge \frac{|E_i|}{d|E_i|(1 + w_i)} \ge \frac{1}{d(1 + 6d/\varepsilon_0)} = q.
\]
Substituting $g_i \ge q$ into \eqref{eq:ind-step} gives
\[
\Pr\left[\cE^*_{\leq M} \mid \cH_{i - 1} \right] \ge f(i + 1, \ell) + r_i \Big(q \cdot f(i + 1, \ell - 1) - f(i + 1, \ell)\Big).
\]
By definition of $f(\cdot, \cdot)$, we have $q \cdot f(i + 1, \ell - 1) - f(i + 1, \ell) \geq 0$.
Moreover, we have
\[
r_i = \frac{\sum_{u \in S_i^* \cup F_i} (\deg_{E_i}(u) + w_i)}{\sum_{u \in V_i} (\deg_{E_i}(u) + w_i)} = \frac{\sum_{e \in E_i} |e| + (\ell + f_i) w_i}{\sum_{e \in E_i} |e| + (n-i+1) w_i} \ge \frac{\ell w_i}{(n-i+1) w_i} = \frac{\ell}{n-i+1}.
\]
Substituting these into the above, we obtain
\begin{align*}
    \Pr\left[\cE^*_{\leq M} \mid \cH_{i - 1} \right] 
    &\ge f(i + 1, \ell) + \frac{\ell}{n-i+1} \Big( q \cdot f(i + 1, \ell - 1) - f(i + 1, \ell) \Big) \\
    &= \frac{\ell}{n-i+1} \cdot q \cdot f(i + 1, \ell - 1) + \Paren{1 - \frac{\ell}{n-i+1}} \cdot f(i + 1, \ell) \\
    &= q^\ell \left( \frac{\ell}{n-i+1} \cdot \frac{\binom{n-i - N_0}{\ell - 1}}{\binom{n-i}{\ell - 1}} + \frac{n-i+1 - \ell}{n-i+1} \cdot \frac{\binom{n-i - N_0}{\ell}}{\binom{n-i}{\ell}} \right) \\
    &= q^\ell \left( \frac{\binom{n-i - N_0}{\ell - 1} + \binom{n-i - N_0}{\ell}}{\binom{n-i+1}{\ell}} \right) \\ 
    &= q^\ell \cdot \frac{\binom{n-i+1- N_0}{\ell}}{\binom{n-i+1}{\ell}} = f(i, \ell),
\end{align*}
which completes the induction.
\end{proof}

From the above claim, we get $$\Pr[\cE^*] \geq q^k \cdot \frac{\binom{\lfloor n/2 \rfloor}{k}}{\binom{n}{k}} q^k \prod_{j=0}^{k-1} \frac{\lfloor n/2 \rfloor - j}{n - j} \ge \left(\frac{q}{4}\right)^k = (d + 1/\eps_0)^{-O(k)},$$ where we use the assumption $n \ge 3k$ above to deduce $\frac{\lfloor n/2 \rfloor - j}{n - j} \geq \frac{1}{4}$.
\end{proof}

We remark that, the intuition behind the induction proof is relatively simple. For $\cE^*$ to occur, it suffices for the following two events to happen: (i) all $k$ relevant decisions are made before time $M$, and (ii) each of the $k$ relevant decisions turns out to be good. The probability that (i) occurs is no worse than the probability that a random permutation on $n$ items contains $k$ marked items among the first half, which is $2^{-O(k)}$. Meanwhile, as argued through $g_i \geq q$ above, the probability that (ii) occurs for each decision is at least $q$. Combining these two yields the desired bound on $\Pr[\cE^*]$. Note that we chose to use induction above to avoid complicated coupling arguments.

\subsubsection{High Probability via Private Selection}

Next, we show that we can use private selection from private candidates \cite{LT19,Papernot022,GhaziKKKMZ25} to boost the probability. Specifically, we wish to run the algorithm to produce $\pi$, compute the resulting hitting set size $|S_{\pi}(E)|$, and only output if this is below $k$. However, doing this directly would not be DP. To achieve DP, we follow the framework of \cite{GhaziGKKKM26}, where each run is also randomly dropped. The full description is presented in \Cref{alg:dp-fpt-dhs-enc}.

\begin{algorithm}[H]
\caption{Encoder for $d$-Hitting Set}
\label{alg:dp-fpt-dhs-enc}
\begin{algorithmic}[1]
\REQUIRE Hypergraph $G = (V, E)$ where each edge has size at most $d$, privacy parameter $\eps$, target size $k$.
\STATE Let $\eps' \coloneqq \frac{\eps}{3}$.
\STATE Set the score threshold: $\tau \coloneqq -k$.
\STATE Sample $\gamma \sim \mathrm{Geom}(e^{-\eps'}).$
\STATE Let $n' \gets n(1 + \eps)$ and $T \gets \lceil 2n'\ln(n') \cdot (d + 1/\eps)^{C \cdot k} \rceil$ \hfill \COMMENT{$C$ is a sufficiently large constant}
\FOR{$i = 1, \dots, T$}
\STATE Sample $y_i \sim \mathrm{Ber}(e^{-\eps' \cdot \gamma}).$
\IF{$y_i = 1$}
\STATE $\pi_i \gets$ \textsc{DP-FPT-dHS-Gen}$(G; \eps'/2)$.
\STATE Sample $\eta_i \sim \mathrm{Lap}(2/\eps')$.
\IF{$-|S_{\pi_i}(E)| + \eta_i \ge \tau$}
\RETURN $\pi_i$.
\ENDIF
\ENDIF
\ENDFOR
\RETURN an arbitrary default permutation $\pi_{\mathrm{def}}$.
\end{algorithmic}
\end{algorithm}

We can now prove our main theorem of this subsection.

\begin{proof}[Proof of \Cref{thm:vc-main}]
We use \Cref{alg:dp-fpt-dhs-enc} as the encoder. It is clear that the algorithm runs in $(d + 1/\eps)^{O(k)} \cdot n^{O(1)}$ time. We next analyze its privacy and utility.

\textbf{Privacy:}
For a fixed $i \in [T]$, consider the mechanism $\cM_i$ that computes $(\pi_i, -|S_{\pi_i}(E)| + \eta_i)$. For a fixed $\pi_i$, $|S_{\pi_i}(E)|$ has sensitivity one. Thus, $\cM_i$ can be viewed as a composition of the $(\eps'/2)$-DP \textsc{DP-FPT-dHS-Gen} and the $(\eps'/2)$-DP Laplace mechanism, meaning that $\cM_i$ is $\eps'$-DP. Since \Cref{alg:dp-fpt-dhs-enc} is simply an instantiation of \cite[Algorithm 4]{GhaziKKKMZ25} with these mechanisms $\cM_i$, \cite[Theorem 14]{GhaziKKKMZ25} implies that \Cref{alg:dp-fpt-dhs-enc} is $(3\eps')$-DP.

\textbf{Utility:}
Let $\Edef$ denote the event that the algorithm reaches the last line and outputs $\pi_{\mathrm{def}}$, and let $X \coloneqq \max(0, |S_{\pi}(E)| - k)$, where $\pi$ is the output.

Let $S^*$ be an optimal hitting set of size (at most) $k$. By Lemma~\ref{lem:dhs-gen}, each candidate run generates a permutation $\pi_i$ satisfying $S_{\pi_i}(E) = S^*$ and $|S_{\pi_i}(E)| \le k$ with probability at least $p = (d + 1/\eps)^{-C \cdot k}$ for some constant $C > 0$. (We use this same constant when setting the value of $T$.)

If $y_i = 1$ and this event occurs, $\pi_i$ is returned with probability at least 1/2 (because $\Pr[\eta_i \geq 0] = 1/2$). 
Thus, conditioned on $y_i = 1$, the probability of generating and returning $\pi_i$ with $|S_{\pi_i}(E)| \le k$ is at least $p_0 \coloneqq p/2$.

For any $\Delta > 0$, a candidate permutation $\pi_i$ is $\Delta$-suboptimal if $|S_{\pi_i}(E)| > k + \Delta$. Conditioned on $y_i = 1$, the probability of generating and accepting such a candidate is at most:
\[
p_{\text{bad}}(\Delta) \le \Pr\left[ \eta_i \ge \Delta\right] = \frac{1}{2} e^{-\Delta \eps / 6}.
\] 
The algorithm accepts the first candidate whose noisy score exceeds the threshold. Conditioned on accepting some candidate (i.e., $\neg \Edef$), the probability that the accepted candidate is $\Delta$-suboptimal is bounded by the ratio of accepting a bad candidate to accepting a good candidate in any single round:
\begin{align*} 
\Pr[X \geq \Delta \mid \neg \Edef] \le \frac{p_{\text{bad}}(\Delta)}{p_0} \le (d + 1/\eps)^{C \cdot k} e^{-\Delta \eps / 6}.
\end{align*}
Let $\Delta_0 := \Theta\left(\frac{k \log(d + 1/\eps)}{\eps}\right)$ be the value of $\Delta$ such that the RHS above is equal to one.

Finally, we bound $\Pr[\Edef]$ as follows. First, let $\gamma^* = \left\lfloor \frac{\ln(n')}{\eps'} \right\rfloor$.
\begin{align*}
\Pr[\Edef] &\leq \Pr[\Edef \mid \gamma \leq \gamma^*] + \Pr[\gamma > \gamma^*] \\
&\leq \prod_{i=1}^T \Paren{1 - \Pr[y_i = 1 \mbox{ and } |S_{\pi_i}(E)| \leq k \mbox{ and } \eta_i \geq 0 \mid \gamma \leq \gamma^*]} + \frac{1}{n'} \\
&\leq \Paren{1 - \frac{1}{n'} \cdot p \cdot \frac{1}{2}}^T + \frac{1}{n'} 
\leq \exp\Paren{-T \cdot \frac{p}{2n'}} + \frac{1}{n'} \leq \frac{2}{n'}.
\end{align*}

Combining the above two inequalities, we get
\begin{align*}
\E[X] &\leq \E[X \mid \neg \Edef] + n \cdot \Pr[\Edef] \\
&\leq  \Paren{\Delta_0 + \int_{\Delta_0}^{n} e^{-(\Delta - \Delta_0) \eps / 6} d\Delta} +  n \cdot \frac{2}{n'} \leq \Delta_0 + \frac{6}{\eps} + \frac{2}{1 + \eps} \leq O\Paren{\frac{k \log(d + 1/\eps)}{\eps}}.
\end{align*}
Thus, the expected size of the final hitting set returned by the decoder is $\E[|S_{\pi}(E)|] \le k + O(\frac{k \log(d+1/\eps)}{\eps}) = k(1 + O(\frac{\log(d+1/\eps)}{\eps}))$. This yields the claimed $(1 + O(\frac{\log(d+1/\eps)}{\eps}))$-approximation.
\end{proof}
\subsection{Subgraph Packing via Color Coding}
\label{sec:color-coding}

In this section, we present the Color Coding technique for $H$-packing. The final guarantee of our algorithm is as follows.

\begin{theorem} \label{thm:color-coding}
Let $H$ be any fixed hypergraph. For $\eps \geq 1$, there is an $\eps$-DP $\Paren{1 - O\Paren{\frac{\log k + \frac{\log \log (1+\eps)}{k}}{\eps}}}$-FPT-approximation implicit representation scheme for $H$-Packing with running time $k^{O(k)} \cdot n^{O(1)} \cdot \ln(1 + \eps)$.
\end{theorem}

We note that our approximation guarantee is only non-trivial when $\eps \geq \Omega(\log k)$. While this might seem restrictive, we show below (\Cref{subsec:color_coding_lb}) that this dependency is necessary if we work in the current implicit representation.

Recall that a hyperedge $e \in E$ is called \emph{monochromatic} under a vertex coloring $c: V \to [k]$ if all vertices in $e$ are assigned the same color. Our coloring representation can be formalized as follows.

\begin{definition}[Coloring Representation] \label{def:coloring-rep}
In the \emph{coloring representation}, the encoder outputs a vertex coloring $c: V \to [k]$ as the implicit representation. 

The decoder $\cD(G, c)$ operates by iterating through each color $i \in [k]$ and selecting a monochromatic copy of $H$ of color $i$ if one exists.
\end{definition}

The algorithm is very simple: The encoder simply use the exponential mechanism (\Cref{def:em}) to pick the best coloring (from random colorings), as formalized below.

\begin{algorithm}[H]
\caption{Color Coding Encoder for $H$-Packing}
\label{alg:color-coding}
\begin{algorithmic}[1]
\REQUIRE Hypergraph $G = (V, E)$, target size $k$, privacy budget $\eps$
\STATE $t \gets \lceil k^{q k} \cdot 2 \ln\Paren{k + \eps} \rceil$, where $q$ is the number of vertices in $H$
\STATE Candidate space $\mathcal{C}$ consists of $t$ uniformly random colorings $c: V \to [k]$
\STATE Score function $\scr(G, c) :=$ number of colors $i \in [k]$ with a monochromatic copy of $H$ w.r.t. $c$
\RETURN $\EM_{\scr}(G)$
\end{algorithmic}
\end{algorithm}

We can now easily prove our main theorem.

\begin{proof}[Proof of \Cref{thm:color-coding}]
%
We use \Cref{alg:color-coding}; 
let $\tilde{c}$ be the
output.  Since the score of each coloring can be computed in polynomial time, the total running time is $t \cdot N^{O(1)} = k^{O(k)} \cdot n^{O(1)}$.  

\textbf{Privacy:} Since $\scr(G, \cdot)$ has sensitivity one, \Cref{thm:em} implies that \Cref{alg:color-coding} is $\eps$-DP.  

\textbf{Utility:} Assume $G$ contains a packing $\mathcal{P}^*$ of $k$ disjoint subsets $S_1, \dots, S_k$ where each $G[S_i]$ is isomorphic to $H$. We say that a coloring $c: V \to [k]$ is \emph{good} if each of $S_1, \dots, S_k$ is monochromatic and they all get assigned different colors. A random coloring is good with probability at least $k^{-qk}$. Let $\Egood$ denote the event that at least one coloring in $\cC$ is good. When $\Egood$ holds, we have $\max_{c \in \cC} \scr(G, c) = k$. Thus, from the previous argument and \Cref{thm:em}, we have
\begin{align*}
\E[\scr(G, \tc)] &\geq \E[\scr(G, \tc) \mid \Egood] \cdot \Pr[\Egood] \\
&\geq \Paren{k - O\Paren{\frac{\log t}{\eps}}} \cdot \Paren{1 - (1 - k^{-qk})^t} \\
&\geq \Paren{k - O\Paren{\frac{k \log k + \log \log(1 + \eps)}{\eps}}} \cdot \Paren{1 - \frac{1}{(k + \eps)^2}} \\
&= k -  O\Paren{\frac{k \log k +  \log \log (1+\eps)}{\eps}}.
\end{align*}
Since the decoder finds exactly one copy of $H$ for each satisfied color $i \in [k]$, the number of returned subgraphs matches the score exactly. Thus, this yields an $\Paren{1 - O\Paren{\frac{\log k + \frac{\log \log (1+\eps)}{k}}{\eps}}}$-approximation.
\end{proof}

\subsection{Maximum Induced Subgraph with Hereditary Property}
\label{sec:spatial-restriction}

In this section, we present the subset representation, tailored for finding maximum induced subgraphs satisfying a hereditary property $\Gamma$ on graph classes satisfying two mild properties (\Cref{assum:large-ind} and \Cref{assum:constant-approx-deletion}), which will be discussed in more detail below. 
Under these assumptions, we provide $(1 - 1/\eps)$-approximation under our implicit representation framework.

\begin{theorem} \label{thm:spatial-kernel}
Let $\Gamma, \cG$ be any hereditary property and any hereditary graph class satisfying \Cref{assum:large-ind} and \Cref{assum:constant-approx-deletion}. For $\eps \geq 1$, there is an $\eps$-DP $\Paren{1 - O\Paren{1/\eps}}$-FPT-approximation implicit representation scheme for Maximum Induced $\Gamma$-subgraph with running time $2^{O(k)} \cdot n^{O(1)}$.
\end{theorem}
As mentioned above, we require two natural assumptions that will be satisfied by many sparse graph classes. First, we assume that there always exists a linear-size subgraph satisfying $\Gamma$ and that such a subgraph can be found in polynomial time.
\begin{assumption}[Large Baseline]
\label{assum:large-ind}
There is a polynomial-time algorithm $\algbaseline$ that, given any graph $G = (V, E) \in \cG$, can find $S \subseteq V$ of size at least $|V| / C_{\cG}$ such that $G[S]$ satisfies $\Gamma$, where $C_{\cG} > 0$ is a constant.
\end{assumption}

To state the second assumption, it will be more convenient to first define the ``dual'' optimization problem to the Maximum $\Gamma$-Subgraph:
\begin{definition}[Deletion to $\Gamma$]
In the $\Gamma$-Deletion problem, we are given $G = (V, E)$ and the goal is to output $S$ of a smallest size such that $G[V \setminus S]$ satisfies $\Gamma$; let $\dist_{\Gamma}(G)$ denote the optimum.
\end{definition}
We assume there is a constant approximation for this problem on the  graph classes of interest:
\begin{assumption}[Constant Approximation for Deletion to $\Gamma$] \label{assum:constant-approx-deletion}
For some constant $\rho \geq 1$, there exist a polynomial-time algorithm $\algdeletion$ that, for any given graph $G \in \cG$, outputs $S'$
such that $G[V \setminus S']$ satisfies $\Gamma$ and $|S'| \leq
\rho \cdot \dist_{\Gamma}(G)$.
\end{assumption}

We remark that the algorithms in \Cref{assum:large-ind} and \Cref{assum:constant-approx-deletion} can be non-private.
We can now define our subset representation.

\begin{definition}[Subset Representation] \label{def:subset-rep}
In the \emph{subset representation}, the encoder outputs a subset $T \subseteq V$ of size $k$ as the implicit representation.

The decoder $\cD(G = (V, E), T)$ works as follows:
\begin{itemize}[nosep]
\item If $|V| \geq k \cdot C_{\cG}$, it simply returns the output from $\algbaseline(G)$ from \Cref{assum:large-ind}.
\item Otherwise, it return $T \setminus \algdeletion(G[T])$ where $\algdeletion$ is from \Cref{assum:constant-approx-deletion}.
\end{itemize}
\end{definition}

Note that the two assumptions imply that the decoder runs in polynomial time. 
We call this the ``subset'' representation because, in the non-trivial case $|V| < k \cdot C_\cG$, the output of $T \setminus \algdeletion(G[T])$ is always a subset of the implicit representation $T$.

\subsubsection{Algorithm}

We use the exponential mechanism to select the set, where the score of each subset $T$ is (a scaled version of) the output size guarantee of $\algdeletion(G[T])$, using \Cref{assum:constant-approx-deletion}.

\begin{algorithm}[H]
\caption{Algorithm for Property $\Gamma$}
\label{alg:spatial-kernel}
\begin{algorithmic}[1]
\REQUIRE Graph $G = (V, E)$, target size $k$, privacy parameter $\eps$, hereditary property $\Gamma$, graph class constant $C_{\cG}$ from \Cref{assum:large-ind}, approximation ratio $\rho$ from \Cref{assum:constant-approx-deletion}.
\ENSURE Implicit representation $T \subseteq V$
\IF{$|V| \geq k \cdot C_{\cG}$}
\RETURN Arbitrary size-$k$ subset of $V$
\ENDIF
\STATE  Candidate space $\cC \gets \binom{V}{k}$ \hfill \COMMENT{All size-$k$ subsets of $V$}
\STATE Score function $\scr(G, T) \gets \frac{|T|}{\rho} - \dist_\Gamma(G[T])$ \hfill \COMMENT{Compute $\dist_\Gamma(G[T])$ using brute-force}
\RETURN $\EM_{\scr}(G)$
\end{algorithmic}
\end{algorithm}

\begin{proof}[Proof of \Cref{thm:spatial-kernel}]
We use \Cref{alg:spatial-kernel}. If $|V| \geq k \cdot C_{\cG}$, then the algorithm immediately returns and thus runs in polynomial time. Otherwise, since the score of each subset can be computed via brute-force in $2^{O(|V|)} = 2^{O(k)}$ time, the total running time is $2^{O(k)} \cdot n^{O(1)}$.

\textbf{Privacy:} 
By \Cref{thm:em}, it suffices to show that, for each $T \subseteq V$,  $\scr(G, T)$ has sensitivity one. Since $|T|/\rho$ is fixed, it is in turn sufficient to show that $\dist_{\Gamma}(G[T])$ has sensitivity one.

To see that this holds, consider neighboring graphs $G = (V, E)$ and $G' = (V, E \cup \{e\})$ where $e = \{u, v\}$. For any deletion set $D \subseteq T$, we always have that $G[T \setminus (D \cup \{u\})] = G'[T \setminus (D \cup \{u\})]$. In other words, if $G[T \setminus D]$ (resp., $G'[T \setminus D]$) satisfies $\Gamma$, then $G'[T \setminus (D \cup \{u\})]$ (resp., $G[T \setminus (D \cup \{u\})]$) satisfies $\Gamma$.  Thus, we have $|\dist_{\Gamma}(G[T]) - \dist_{\Gamma}(G'[T])| \leq 1$, yielding the desired sensitivity guarantee.

\textbf{Utility:} Assumption~\ref{assum:large-ind} immediately handles  the case $|V| \geq k \cdot C_{\cG}$. We henceforth consider the remaining case $|V| < k \cdot C_{\cG}$. Let $\tT$ denote the output of \Cref{alg:spatial-kernel}.
If there exists $S^*$ such that $|S^*| \geq k$ and $G[S^*]$ satisfies $\Gamma$, then by
the hereditary property, there is a
subset $T^* \subseteq S^*$ such that
$|T^*| = k$ in $\cC$ with $\scr(G, T^*) \geq k/\rho$. Thus, by the guarantee of the exponential mechanism (\Cref{thm:em}), we have
$\E\left[\scr(G, \tT)\right] \geq k/\rho - O(n/\eps).$ Rearranging this gives 
\begin{align*}
\E[|\tT| - \rho \cdot \dist_\Gamma(G[\tT])] \geq k - O(\rho \cdot n/\eps) \geq k\Paren{1 - O_{\rho, C_{\cG}}\Paren{\frac{1}{\eps}}}.
\end{align*}
Finally, recall from \Cref{assum:constant-approx-deletion} that the output from $\algdeletion(G[\tT])$ has size at most $\rho \cdot \dist_\Gamma(G[\tT])$. Hence, the expected output size from the decoder is at least $k(1 - O(1/\eps))$.
\end{proof}

\subsubsection{Applications}

We now apply Theorem~\ref{thm:spatial-kernel} to several important combinations of problems and degenerate graphs. To do this, we first make a simple observation that \Cref{assum:large-ind} holds quite generally:
\begin{observation} \label{obs:large-ind}
Let $\Gamma$ be any hereditary property that contain all independent sets (i.e., empty graphs), and $\cG$ be the class of all $\Delta$-degenerate graphs. Then, \Cref{assum:large-ind} holds with $C_{\cG} = \Delta + 1$.
\end{observation}

\begin{proof}
Here the baseline algorithm $\algbaseline$ is simply the greedy algorithm: Pick a vertex with the minimum degree, remove it and all its neighbors, and repeat this process until the graph becomes empty. By $\Delta$-degeneracy, we remove at most $\Delta + 1$ vertices in each iteration. Therefore, we find an independent set of size at least $\frac{|V|}{\Delta + 1}$.
\end{proof}
This immediately yields the following.
\begin{corollary}
Let $\Delta \in \N$ be any constant, and $\Gamma$ be any of the following hereditary properties:
\begin{enumerate}[nosep]
\item Independent Sets (i.e., empty graphs), \label{item:first-forbidden}
\item Cluster Graphs,
\item Claw-Free Graphs,
\item Split Graphs,
\item Cographs,
\item Bounded Degree, \label{item:last-forbidden}
\item Forests (i.e., Cycle-free graphs),
\end{enumerate}
Then, for any $\eps \geq 1$, there is an $\eps$-DP $\Paren{1 - O\Paren{1/\eps}}$-FPT-approximation implicit representation scheme for Maximum Induced $\Gamma$-subgraph with running time $2^{O(k)} \cdot n^{O(1)}$ on $\Delta$-degenerate graphs.
\end{corollary}

\begin{proof}
From \Cref{thm:spatial-kernel}, it suffices to check that \Cref{assum:large-ind} and \Cref{assum:constant-approx-deletion} hold. The former follows immediately from \Cref{obs:large-ind}. For the latter, the hereditary properties in items \ref{item:first-forbidden}-- \ref{item:last-forbidden} can be formulated as forbidden (induced) subgraphs with bounded sizes. Such problems always admit $O(1)$-approximation algorithm in polynomial time \cite{LundY93}. 

For item 7 (i.e., Forests), the associated deletion problem is the Feedback Vertex Set, which admits 2-approximation algorithm in polynomial time~\cite{BafnaBF99}.
\end{proof}
\section{Representation-Independent (Computational) Lower Bounds}
\label{sec:lower-bounds}

While we have provided algorithms in the previous section, a natural question is whether we can significantly improve the approximation-privacy tradeoff by, e.g., coming up with more innovative representation. In this section, we partially answer this question in the negative: We establish a representation-independent lower bound that applies to any implicit representation whose decoder runs in polynomial (or even subexponential\footnote{We remark that, if we instead allow the decoder exponential time, then it can simply use brute-force to solve the problems without even the need of the encoder at all.}) time. This holds for any sufficiently small (constant) privacy budget and approximation ratio, as stated more formally below.



\begin{theorem}
\label{thm:main-rep-ind-lb}
Assuming Gap-ETH with Advice, there exist constants $\eps > 0$ and $d \in \N$, $\alpha_{\text{VC}} > 1$, and $\alpha_{\text{IS}}, \alpha_{\text{TP}} < 1$ such that no $\eps$-DP implicit representation scheme with a subexponential-time decoder can achieve an $\alpha_{\text{VC}}$-approximation for Vertex Cover, an $\alpha_{\text{IS}}$-approximation for Maximum Independent Set, or an $\alpha_{\text{TP}}$-approximation for Triangle Packing on $d$-bounded-degree graphs.
\end{theorem}

In fact, our lower bounds hold very generally for any problems which are ``Gap-3SAT-hard'', which roughly means that there exists a gap-preserving reduction from Gap-3SAT with linear size. (See \Cref{def:gap-3sat-hard} and \Cref{lem:lb-gap-3sat-hard} below for a more precise formulation.)

Crucially, the lower bound in \Cref{thm:main-rep-ind-lb} holds without any restriction on the running time of the encoder $\cE$. This means that we cannot simply run the encoder to arrive at a contradiction.

\subsection{From Privacy to Advice String}

Instead, what we show is that we can extract a small advice string from the encoder using only its privacy guarantee (\Cref{claim:advice-size}). This allows us to prove the following generic lemma, which turns any implicit representation scheme (with sufficiently small $\eps$), to a (non-private) algorithm that solves the problem using mildly exponential time and advice length.

\begin{lemma}[Main Lemma] \label{lem:ind-lb}
Let $\Pi$ be any optimization problem on graphs. 
Let $d \in \N$ and $\delta, \rho > 0$ be constants. Then, there exists a constant $\eps > 0$ (depending on $d, \delta$) such that the following holds.

If there is an $\eps$-DP implicit representation scheme which solves $\rho$-gap-$\Pi$ with one-sided error on all $d$-bounded-degree graphs such that the decoder runs in $2^{o(n)}$ time, then there is an $O(2^{\delta n})$-time deterministic algorithm with advice length $O(2^{\delta n})$ that solves $\rho$-gap-$\Pi$ on all $d$-bounded-degree graphs.
\end{lemma}

\begin{proof}
We use the notation for maximization problem $\Pi$; the minimization setting works similarly.

Let $\eps = \delta/d$, and let $\cG_{n, d}$ be the set of all $n$-vertex graphs with $d$-bounded degree. 

Assume, for the sake of contradiction, that there exists an $\eps$-DP implicit representation scheme $(\cE, \cD)$ that solves $\rho$-gap-$\Pi$ for all $G \in \cG_{n, d}$ such that $\cD$ runs in $2^{o(n)}$ time. Let $\cC$ denote the set of all possible representations.

Let $M$ denote the number of random bits used by $\cD$. For clarity, we write a subscript $\cD_r$ to denote that $\cD$ is run using randomness $r \in \{0, 1\}^M$. 

The main claim is the following which guarantees the existence of the advice set.

\begin{claim} \label{claim:advice-size}
There exists a set $A_{n, k}$ of pairs $(r, C) \in \{0, 1\}^M \times \cC$ such that the following holds.
\begin{itemize}[nosep]
\item (Acceptance Guarantee) For every $G \in \cG_{n, d}$ with $\opt(G) = k$, there exists $(r^*, C^*) \in A_{n, k}$ such that $\cD_{r^*}(G, C^*) = $ YES.
\item (Rejection Guarantee) For every $G \in \cG_{n, d}$ with $\opt(G) < \rho \cdot k$ and for all $(r, C) \in A_{n, k}$, we have $\cD_{r}(G, C) = $ NO.
\item (Size Guarantee) $|A_{n, k}| \leq e^{(\eps d / 2 + o(1)) n}$.
\end{itemize}
\end{claim}

\begin{proof}[Proof of \Cref{claim:advice-size}]
From the assumption of $(\cE, \cD)$, for some constant $\xi > 0$, we have
\begin{align}
&\Pr_{r \sim \{0, 1\}^M, \cE}[\cD_r(G, \cE(G)) \text{ = } \mathit{YES}] \geq \xi &&\forall G \in \cG_{n, d} \text{ s.t. } \opt(G) = k, \label{eq:rejection} 
\end{align}
Let $G_\emptyset$ be the empty graph on $n$ vertices.
Since each $G$ has at most $nd/2$ edges, we may apply group privacy to conclude that, for any event $\cO$, we have
\begin{align*}
\Pr[\cE(G) \in \cO] \leq e^{\eps n d / 2} \cdot \Pr[\cE(G_\emptyset) \in \cO] 
\end{align*}
Applying this to \eqref{eq:rejection}
, we arrive at
\begin{align}
&\Pr_{r \sim \{0, 1\}^M, \cE}[\cD_r(G, \cE(G_\emptyset)) \text{ = } \mathit{YES}] \geq \frac{\xi}{e^{\eps n d / 2}} &&\forall G \in \cG_{n, d} \text{ s.t. } \opt(G) = k, \label{eq:rejection-empty} 
\end{align}
Let $A_{n, k}$ be the set of $B := \left\lceil \frac{e^{\eps n d / 2}}{\xi} \right\rceil \cdot n^2$ independent pairs $(r, C)$ where $r$ is drawn uniformly at random from $\{0, 1\}^M$ and $C$ from $\cE(G_\emptyset)$. Obviously, this satisfies the size guarantee. The rejection guarantee follows from the fact that the algorithm has one-sided error.

Finally, we will show that it satisfies the acceptance guarantee with non-zero probability, which would conclude our proof. To see that this is true, consider any fixed $G \in \cG_{n, d}$ with $\opt(G) = k$. From \eqref{eq:rejection-empty}, we have
\begin{align*}
\Pr[\exists (r^*, C^*) \in A_{n, k}, \cD_{r^*}(G, C^*) = \mathit{ YES}] &\geq 1 - \Paren{1 - \frac{\xi}{e^{\eps n d / 2}}}^B \geq 1 - e^{-B \cdot \frac{\xi}{e^{\eps n d / 2}}} \geq 1 - e^{-n^2}.
\end{align*}
Thus, by the union bound over all $G \in \cG_{n, d}$ (using the trivial bound $|\cG_{n, d}| \leq 2^{n^2}$), the probability that the acceptance guarantee is satisfied is at least 
\begin{align*}
1 - e^{-n^2} \cdot |\cG_{n, d}| \geq 1 - e^{-n^2} \cdot 2^{n^2} > 0,
\end{align*}
which concludes our proof.
\end{proof}

Our deterministic algorithm for solving $\rho$-gap-$\Pi$ works as follows:
\begin{itemize}[nosep]
\item Let the advice string be the union of $A_{n, k}$ for all valid $k \leq n^{O(1)}$.
\item For every $(r, C) \in A_{n, k}$, we run the algorithm $\cD_r(G, C)$.
\item If any of these runs returns a YES, then outputs YES.
\item Otherwise, output NO.
\end{itemize}
The correctness immediately follows from the acceptance and rejection guarantees from \Cref{claim:advice-size}. The advice size is $|A_{n, k}| \cdot 2^{o(n)} \leq e^{(\eps d / 2 + o(1)) n} \leq O(2^{\delta n})$. Finally, since we assume that $\cD$ runs in $2^{o(n)}$ time, the total running time of the algorithm is $|A_{n, k}| \cdot 2^{o(n)} \leq O(2^{\delta n})$.
\end{proof}

\subsection{Lower Bounds for Gap-3SAT-hard Problems}


With \Cref{lem:ind-lb} in mind, we can derive a contradiction with Gap-ETH with Advice by simply establishing a reduction from Gap-3SAT (\Cref{def:gap-3sat}). To provide a unified treatment for all the problems, we state the requirements for such reductions below.

\begin{definition}[Gap-3SAT-hard] \label{def:gap-3sat-hard}
We say that a maximization (resp. minimization) problem $\Pi$ on graphs is Gap-3SAT-hard if, for every constants $\gamma \in (0, 1), D \ge 3$, there is a polynomial-time reduction that takes in a $\gamma$-Gap-3SAT($D$) instance  $\phi$ and returns an instance $(G, k)$ of $\Pi$ such that the following holds for some constants $C, d > 0$ and $\rho < 1$ (resp. $\rho > 1$) depending only on $\gamma, D$.
\begin{itemize}[nosep]
\item (Completeness) If $\phi$ is satisfiable, then $\opt(G) = k$.
\item (Soundness) If $\phi$ is not $\gamma$-satisfiable, then $\opt(G) < \rho \cdot k$ (resp. $\opt(G) > \rho \cdot k$).
\item (Size Bound) The number of vertices of $G$ is $\leq CN$ where $N$ is the number of variables in $\phi$.
\item (Degree Bound) $G$ is $d$-bounded-degree.
\end{itemize}
\end{definition}

Many of the hardness of approximation reductions already satisfy these properties. Indeed, it is well known that the three problems we study are hard in this sense:

\begin{lemma} \label{thm:all-prob-3sat-hard}
Maximum Independent Set, Vertex Cover and Triangle Packing are Gap-3SAT-hard.
\end{lemma}

\begin{proof}[Proof Sketch]
\begin{itemize}
\item \textbf{Maximum Independent Set:} We use the FGLSS reduction~\cite{FeigeGLSS96}. Given a 3-CNF formula $\phi$, we construct a graph $G$ as follows. For each clause $C_j$, we create 7 vertices, each corresponding to one of the 7 satisfying assignments of $C_j$. We add edges between any two vertices if they assign conflicting truth values to a shared variable. It is simple to see that the optimal independent set size is exactly the same as the maximum number of clauses satisfied by any assignment. Since every variable in the 3-CNF formula appears in at most $D$ clauses, the number of vertices in $G$ is at most $O(ND)$ and then the resulting graph is $(21D)$-degree-bounded. 
\item \textbf{Vertex Cover:} This follows directly from the above reduction by recalling that, in any graph $G = (V, E)$, a set $S$ is a vertex cover iff $V \setminus S$ is an independent set. 
\item \textbf{Triangle Packing:} We first apply Schaefer's classic reduction from 3SAT to 1-in-3SAT~\cite{Schaefer78}, and then apply the reduction by van Rooij et al.~\cite{RooijNB13} to Triangle Packing on 4-bounded-degree graphs. The completeness and soundness were proved in \cite[Lemma 27]{GuptaLLMW19}. Since both reductions simply apply local gadgets, the size bound also holds. \qedhere
\end{itemize}
\end{proof}

As mentioned earlier, such a reduction, together with \Cref{lem:ind-lb}, immediately implies a Gap-ETH-based lower bound against implicit representation schemes for $\Pi$, as formalized below.

\begin{lemma} \label{lem:lb-gap-3sat-hard}
Let $\Pi$ be any maximization (resp. minimization) problem on graphs that is Gap-3SAT-hard. Then, assuming Gap-ETH with Advice, there exist constants $\eps > 0$ and $d \in \N$ and $\alpha < 1$ (resp. $\alpha > 1$) such that no $\eps$-DP implicit representation scheme with a subexponential-time decoder can achieve an $\alpha$-approximation for $\Pi$ even on $d$-bounded-degree graphs.
\end{lemma}

\begin{proof}
Let $\gamma, c_0, D$ be constants from Gap-ETH with Advice (\Cref{assump:gap-eth-advice}), and let $C, d, \rho$ be constants from the Gap-3SAT-hardness reduction (\Cref{def:gap-3sat-hard}). Let $\delta = \frac{c_0}{2C}$, let $\eps$ be as in \Cref{lem:ind-lb}, and let $\alpha$ be any constant in $(\rho, 1)$ (resp. $(1, \rho)$).

Suppose for the sake of contradiction that there exists an $\eps$-DP implicit representation scheme with a $2^{o(n)}$-time decoder that can achieve an $\alpha$-approximation for $\Pi$ even on $d$-bounded-degree graphs. By \Cref{obs:exp-to-prob-approx}, this yields an $\eps$-DP implicit representation scheme with a $2^{o(n)}$-time decoder that solves $\rho$-gap-$\Pi$. From this, \Cref{lem:ind-lb} implies that there exists an $O(2^{\delta n})$-time deterministic algorithm $\alg$ with advice length $O(2^{\delta n})$ that solves $\rho$-gap-$\Pi$ on all $d$-bounded-degree graphs.

We can use this to solve the $\gamma$-Gap-3SAT($D$) problem as follows:
\begin{itemize}[nosep]
\item Given an instance $\phi$ of $N$ variables, run the reduction in \Cref{def:gap-3sat-hard} to arrive at $(G, k)$ where $G$ is $d$-bounded-degree and has $n \leq C \cdot N$ vertices.
\item Then, run the algorithm $\alg$ on $G$.
\item Finally, return the output of $\alg$.
\end{itemize}
The correctness of this algorithm immediately follows from the completeness and soundness in \Cref{def:gap-3sat-hard}. As for the running time and advice string\footnote{For the advice string, we use all advice strings of $\alg$ for all $n \leq C \cdot N$ and all possible values of $k \leq n^{O(1)}$.}, both are $2^{\delta n} \cdot n^{O(1)} \leq O(2^{c_0 N})$. Hence, this violates Gap-ETH with Advice.
\end{proof}

Finally, we remark that \Cref{thm:main-rep-ind-lb} is simply a consequence of \Cref{thm:all-prob-3sat-hard} and \Cref{lem:lb-gap-3sat-hard}.
\section{Representation-Dependent Lower Bounds}
\label{sec:rep-lower-bounds}

While our representation-independent lower bounds in Section~\ref{sec:lower-bounds} are very robust, they only apply to small $\eps$, which is the regime where our algorithms provide weak (or even trivial) guarantees. In this section, we focus instead on representation-dependent lower bounds, where we restrict the format of the implicit representation to be the same as our algorithms in \Cref{sec:alg}. In doing so, 
%
%
we show a more refined privacy-approximation tradeoff which holds even for larger values of $\eps$. 

\subsection{Permutation Representation}

We start with the permutation representation for Vertex Cover and $d$-Hitting Set.
Recall that in the permutation framework (Section~\ref{sec:vc}), the implicit representation is a permutation $\pi$ of the vertices, and the decoder picks the earliest element of each hyperedge according to $\pi$.

For this specific representation, we prove the following lower bound.
 
\begin{theorem} \label{thm:rep-lb-noisy}
Let $d \ge 2$. Any $\eps$-DP implicit representation scheme under the permutation representation (\Cref{def:perm-rep}) for $d$-Hitting Set has approximation ratio at least $(1 + \Omega(e^{-\eps}))$.
\end{theorem}

\begin{proof}
Let $n = (2d - 1) k$. We partition the vertex set $V = [n]$ into $k$ disjoint components $V_1, \dots, V_k$, each of size $2d - 1$. For each component $j \in [k]$, we choose a fixed hyperedge $H_{j, 1} \subseteq V_j$ of size $d$. 

Let $w_j$ be a vertex chosen uniformly at random from $H_{j, 1}$.  Let $H_{j, 2}$ be $\{w_j\} \cup (V_j \setminus H_{j, 1})$; thus, $H_{j, 1} \cap H_{j, 2} = \{w_j\}$. Let $Y \in \{0, 1\}^k$ be a uniformly random binary vector. The hypergraph instance $G_{Y}$ consists of the hyperedges $H_{j, 1}$ for all $j \in [k]$, and additionally contains  $H_{j, 2}$ whenever $Y_j = 1$. Since every hyperedge has size $d$, this is a valid $d$-Hitting Set instance.

For each component $j$, if $Y_j = 0$, the only hyperedge is $H_{j, 1}$, which can be hit by one vertex. If $Y_j = 1$, by construction, both $H_{j, 1}$ and $H_{j, 2}$ can be hit by $\{w_j\}$. Thus, $\opt(G_Y) = k$.

Consider any $\eps$-DP implicit representation scheme $(\cE, \cD)$ under the permutation representation (\Cref{def:perm-rep}).
Let $\pi = \mathcal{E}(G_Y)$ be the output permutation. The decoder picks the earliest vertex from each hyperedge. For component $j$, the decoder picks $u_{j,1} = \min_\pi(H_{j, 1})$ and (if $Y_j = 1$) $u_{j,2} = \min_\pi(H_{j, 2})$. If $u_{j,1} \neq u_{j,2}$, the decoder picks two vertices, incurring an excess cost of 1 for this component. Let $U_j$ be the indicator that $u_{j,1} \neq u_{j,2}$. The total excess cost is $X = \sum_{j=1}^k Y_j U_j$.

Let $Y_{-j}$ be $Y$ with the $j$th coordinate removed.  We will condition on $Y_{-j}$ and the 
random choices of $w$.  Let 
$G_{Y_{-j}, b}$ be the graph when
$Y_j = b$, for $b = 0, 1$; note that $G_{Y_{-j}, 0}$ and 
$G_{Y_{-j}, 1}$ are neighbors.  
Since $\mathcal{E}$ is $\eps$-DP, we have 
\begin{align*}
\E_{Y_j, \mathcal{E}}[ Y_j U_j \mid Y_{-j}, w] 
& = 
\frac{1}{2} 
\E_{\mathcal{E}}[U_j \mid G_{Y_{-j}, 1}, w] 
=
\frac{1}{2} \Pr_{\mathcal{E}}[U_j = 1 \mid G_{Y_{-j}, 1}, w] \\
& \geq \frac{1}{2} e^{-\eps} \Pr_{\mathcal{E}}[U_j = 1 \mid G_{Y_{-j}, 0}, w] 
= \frac{1}{2} e^{-\eps} \E_{\mathcal{E}}[U_j \mid G_{Y_{-j}, 0}, w] \\
& = e^{-\eps} \E_{Y_j, \mathcal{E}}[ (1 - Y_j) U_j \mid Y_{-j}, w]. 
\end{align*}
Taking expectation over
$Y_{-j}, w$, summing over
$j \in [k]$, and applying
the linearity of
expectation,
\begin{equation}
\label{eq:rdlb-perm-1}
\E_{Y, w, \mathcal{E}}[X] \ge e^{-\eps} \E_{Y, w, \mathcal{E}}\left[\sum_{j=1}^k (1 - Y_j) U_j\right].
\end{equation}
We now obtain a lower bound on the RHS.  Indeed, conditioned on $Y_j = 0$, the hyperedge $H_{j, 2}$ is absent from $G_{Y_{-j}, 0}$, meaning the $\mathcal{E}$'s output $\pi$ is independent of the choice of $w_j$ and the random vertices of $H_{j, 2}$. The vertex $u_{j,1} = \min_\pi(H_{j, 1})$ is determined by $\pi$. If $u_{j,1} \neq w_j$, then $u_{j,1} \notin H_{j, 2}$, which guarantees that $u_{j,2} = \min_\pi(H_{j, 2}) \neq u_{j,1}$. Since $w_j$ is chosen uniformly at random from $H_{j, 1}$, we have:
\[
\Pr_{w_j}[U_j = 1 \mid \pi, Y_j = 0] \ge \Pr_{w_j}[u_{j,1} \neq w_j \mid \pi, Y_j = 0] = 1 - \frac{1}{d}.
\]
Taking the expectation over $\pi \sim \mathcal{E}(G_{Y} \mid Y_j = 0)$ gives $\E_{w_j, \mathcal{E}}[U_j \mid Y_j = 0] \ge 1 - 1/d$. Since $Y_j$ is chosen independently, we obtain:
\begin{equation}
\label{eq:rdlb-perm-2}
\E_{Y, w, \mathcal{E}}\left[\sum_{j=1}^k (1 - Y_j) U_j\right] \ge \frac{k}{2} \left( 1 - \frac{1}{d} \right) 
\geq \frac{k}{4},
\end{equation}
for any $d \ge 2$. Combining
\eqref{eq:rdlb-perm-1} and \eqref{eq:rdlb-perm-2}, we obtain $\E_{Y, w, \mathcal{E}}[X] \ge  \frac{k e^{-\eps}}{4}$.
Since $\opt(G_Y) = k$, by averaging, there is an instance $G$ such that  the expected size of the decoded solution is at least $\opt(G) \cdot \Paren{1 + \frac{e^{-\eps}}{4}}$.   
\end{proof}

\subsection{Coloring Representation}
\label{subsec:color_coding_lb}

Recall that in the coloring framework (Section~\ref{sec:color-coding}), the implicit representation is a $k$-coloring $c: V \to [k]$, and the decoder finds one monochromatic copy of $H$ per color class.
Recall from \Cref{alg:color-coding} that $\scr(G, c)$ is the number of colors $i \in [k]$ with at least one monochromatic copy of $H$. Note that $\scr(G, c)$ is exactly the value of the solution returned by the decoder (given the representation $c$).

The following lower bound shows that we need $\eps \geq \Omega(\ln k)$ to achieve a constant factor approximation under this representation, which complements our upper bound (\Cref{thm:color-coding}). 

\begin{theorem} \label{thm:rep-lb-color}
Let $H$ be any fixed connected graph with at least two vertices. There exists a constant $c > 0$ (depending on $H$) such that, for $\eps \leq c \cdot \log k$, any $\eps$-DP implicit representation scheme under the coloring representation (\Cref{def:coloring-rep}) for $H$-Packing has approximation ratio at most $2/3$.
\end{theorem}

\begin{proof}
Let $q$ and $\ell$ denote the number of vertices and the number of edges of $H$, respectively.
Let $n = k^2 q$. We construct a hard distribution $\cP$ over graphs on $V = [n]$. A graph $G \sim \cP$ is formed by choosing $k$ disjoint subsets of $V$ of size $q$ uniformly at random, and planting a complete copy of $H$ on each subset. (The remaining $n - qk$ vertices in $G$ are isolated.) By construction, we have $\opt(G) = k$. 
Moreover, the total number of graphs in the support of $\cP$ is  $M_\mathrm{tot} = \frac{1}{k!} \frac{n!}{(n-qk)! (q!)^k}$.

\begin{claim} \label{claim:col-low-prob}
For any fixed coloring $c: V \to [k]$, we have $\Pr_{G \sim \cP}[\scr(G, c) \ge k/2] \leq O(1/k)^{k/2}$.
\end{claim}

\begin{proof}[Proof of \Cref{claim:col-low-prob}]
For any fixed coloring $c: V \to [k]$, let $\mathcal{G}(c)$ be
the set of graphs $G$ in the support of $\cP$ such that $\scr(G, c) \ge k/2$; we aim to upper bound $|\mathcal{G}(c)|$.
Indeed, to achieve a score of $s \ge k/2$, there must exist $s$ distinct colors, each containing at least one monochromatic copy of $H$. Let $n_i$ be the number of vertices assigned color $i$ by $c$. The number of monochromatic $q$-subsets in color $i$ is $\binom{n_i}{q} \le \frac{n_i^q}{q!}$.
The number of ways to form $s$ disjoint monochromatic copies of $H$ in $s$ distinct colors is bounded by $\frac{1}{(q!)^s} e_s(n_1^q, \dots, n_k^q)$, where $e_s$ is the elementary symmetric polynomial of degree $s$. 
\iffalse
Because the function $x \mapsto x^q$ is strictly convex for $q \ge 2$, the maximum of $e_s(n_1^q, \dots, n_k^q)$ over the simplex $\sum_{i=1}^k n_i = n$, $n_i \ge 0$, is attained at a symmetric boundary point where exactly $r$ variables are non-zero and equal to $n/r$, for some integer $r \in [s, k]$. Evaluating at such a point yields:
\[
e_s(n_1^q, \dots, n_k^q) \le \max_{r \in [s, k]} \binom{r}{s} \left(\frac{n}{r}\right)^{qs} \le \max_{r \ge s} \frac{r^s}{s!} r^{-qs} n^{qs} = \max_{r \ge s} \frac{1}{s!} r^{s(1-q)} n^{qs}.
\]
Since $q \ge 2$, the exponent $s(1-q) \le -s$. Thus, the maximum over $r \ge s$ is attained at $r=s$, strictly bounding the sum by $\frac{1}{s!} s^{-s} n^{qs}$.
\else
For $I$ such that $|I| = s$, let $a_I = \prod_{i \in I} n_i$.  Then, we  have
\begin{align*}
e_s(n_1^q, \dots, n_k^q) 
& = \sum_{|I| = s} a_I^q 
\leq \left(\max_{|I| = s} a_I \right)^{q-1} \sum_{|I| = s} a_I \\
& \leq \left(\frac{n}{s} \right)^{s(q-1)} e_s(n_1, \dots, n_k)
\leq \left(\frac{n}{s} \right)^{s(q-1)} \frac{n^s}{s!} \\
& = \frac{n^{qs}}{s! s^{s(q-1)}} \leq \frac{n^{qs}}{s! s^{s}},
\end{align*}
where the second inequality uses AM--GM on the $s$ coordinates in $I$ 
and the last inequality follows since $q \geq 2$.  
\fi
The remaining $k-s$ copies of $H$ can be placed arbitrarily on the remaining $n - qs$ vertices, which can be done in $\frac{1}{(k-s)!} \frac{(n-qs)!}{(n-qk)! (q!)^{k-s}}$ ways. Thus, the number of graphs satisfying $\scr(G, c) \ge s$ is bounded by:
\[
|\mathcal{G}(c)| \le \left(\frac{n^{qs}}{s! s^s (q!)^s}\right) \frac{(n-qs)!}{(k-s)! (n-qk)! (q!)^{k-s}}.
\]
The fraction of graphs in $\cP$ on which $c$ is successful is therefore bounded by:
\[
\gamma = \frac{|\mathcal{G}(c)|}{M_\mathrm{tot}} \le \frac{k!}{s! (k-s)!} \cdot s^{-s} \cdot \frac{n^{qs} (n-qs)!}{n!} = \binom{k}{s} \cdot s^{-s} \cdot \frac{n^{qs}}{\prod_{j=0}^{qs-1} (n-j)} \le 2^k s^{-s} \left(\frac{n}{n-qs+1}\right)^{qs}.
\]
For $s = k/2$, since $n = k^2 q$ and $k \ge 2$, we have $n - qs + 1 > k^2 q - qk/2 \ge k^2 q / 2$. Thus, $\frac{n}{n-qs+1} \le 2$.
Substituting $s = k/2$ gives:
\begin{align*}
\gamma \le 2^k (k/2)^{-k/2} 2^{qk/2} = \left(\frac{8 \cdot 2^q}{k}\right)^{k/2}. &\qedhere
\end{align*}
\end{proof}

Consider any $\eps$-DP implicit representation scheme $(\cE, \cD)$ under the coloring representation (\Cref{def:coloring-rep}). Let $G_{\emptyset}$ denote the empty graph. By \Cref{claim:col-low-prob} and an averaging argument, there must exist $G^*$ in the support of $\cP$ such that $\Pr_{c \gets \cE(G_\emptyset)}[\scr(G^*, c) \geq k/2] \leq O(1/k)^{k/2}$. This yields
\begin{align*}
\E_{c \gets \cE(G^*)}[\scr(G^*, c)] &\leq \frac{k}{2} + \frac{k}{2} \cdot \Pr_{c \gets \cE(G^*)}[\scr(G^*, c) \geq k/2] \\ &\overset{(\star)}{\leq} \frac{k}{2} + \frac{k}{2} \cdot e^{\eps k \ell} \cdot \Pr_{c \gets \cE(G_\emptyset)}[\scr(G^*, c) \geq k/2] \\ &\leq k \Paren{\frac{1}{2} + O(e^{2\eps\ell} / k)^{k/2}},
\end{align*}
where $(\star)$ follows from group privacy and the fact that $G^*$ contains $k\ell$ edges. Thus, if $\eps < \frac{\log k}{4\ell}$ and $k$ is sufficiently large, the approximation ratio is at most 2/3.
\end{proof}

\subsection{Subset Representation}

Recall that in the subset framework (Section~\ref{sec:spatial-restriction}), the encoder $\mathcal{E}(G)$ outputs a subset $S \subseteq V$ of a fixed size $k$, and the decoder $\cD(G, S)$ returns a subset $I \subseteq S$ satisfying the target property. We prove that any $\eps$-DP algorithm using this framework must incur a loss of $\Omega(e^{-\eps})$ in the approximation ratio. We focus on the Independent Set property.
\begin{theorem} \label{thm:rep-lb-spatial}
Let $\cG$ be the class of one-bounded-degree graphs (i.e. disjoint edges and isolated vertices). Any $\eps$-DP implicit representation scheme under the Subset representation (\Cref{def:subset-rep}) for Independent Set has approximation ratio at most $(1 - \Omega(e^{-\eps}))$.
\end{theorem}
\begin{proof}
Let $n = \lfloor 5k/4 \rfloor$. We define a distribution over graphs on $n$-size vertex set $V$ as follows.
First, draw a uniformly random matching $P$ of size $m = \lfloor k/4 \rfloor$ on $V$. Let $P = \{e_1, \dots, e_m\}$.
Next, draw a uniformly random binary vector $Y \in \{0, 1\}^m$.
The graph $G_{P, Y}$ has edge set $E = \{e_j \mid Y_j = 1\}$. Clearly, $G_{P, Y}$ has maximum degree $1$ and has independent set of size at least $n - m \geq k$.  

We refer to edges $e_j$ with $Y_j = 1$ as \emph{active}, and those with $Y_j = 0$ as \emph{hidden}.

Consider any $\eps$-DP implicit representation scheme $(\cE, \cD)$ under the subset representation (\Cref{def:subset-rep}). Let $S = \mathcal{E}(G_{P, Y})$. Since the extracted set\footnote{Note that the scheme does not revert to the baseline since, for the class $\cG$ of one-bounded-degree graphs, we have $C_{\cG} \geq 2$ (i.e. matching). Thus, $n < k \cdot C_{\cG}$ and we do not use the baseline.} $I_S \subseteq S$ is an independent set, it must omit at least one vertex for every edge of $G_{P, Y}$ contained in $S$. Let $X = \sum_{j=1}^m Y_j \cdot \mathbf{1}[e_j \subseteq S]$ be the number of edges of $G_{P, Y}$ contained in $S$. We have $|I_S| \le k - X$.

By $\eps$-DP, for each $j$, the probability of outputting $S$ decreases by at most a factor of $e^\eps$ if $e_j$ is removed. Thus,
\begin{align*}
\E_{Y, \mathcal{E}} [ Y_j \cdot \mathbf{1}[e_j \subseteq S] ] & = \frac{1}{2} \E_{Y_{-j}} [ \Pr_{\mathcal{E}}[e_j \subseteq S \mid Y_j = 1] ] \\
& \ge \frac{1}{2} e^{-\eps} \E_{Y_{-j}} [ \Pr_{\mathcal{E}}[e_j \subseteq S \mid Y_j = 0] ] = e^{-\eps} \E_{Y, \mathcal{E}} [ (1 - Y_j) \cdot \mathbf{1}[e_j \subseteq S] ].
\end{align*}
Summing over all $j \in [m]$ and taking the expectation over $P$ and $S$, we get:
\[
\E_{P, Y, \mathcal{E}}[X] \ge e^{-\eps} \E_{P, Y, \mathcal{E}}\left[ \sum_{j=1}^m (1 - Y_j) \cdot \mathbf{1}[e_j \subseteq S] \right] = e^{-\eps} \E_{P, Y, \mathcal{E}}[Z],
\]
where $Z$ is the number of hidden edges (those with $Y_j = 0$) contained in $S$.

Conditioned on $G_{P, Y}$, the graph has $L = |\{j \in [m] \mid Y_j = 1\}|$ active edges. The remaining $m - L$ hidden edges of $P$ form a uniformly random matching on the $n - 2L$ vertices outside $G_{P, Y}$. Since $S$ is determined solely by $G_{P, Y}$ and its intersection with the outside vertices has size $k' \ge k - 2L \ge k - 2m \ge k/2$, the expected number of hidden edges contained in these $k'$ vertices is:
\[
\E_{P}[Z \mid G_{P, Y}, S] = (m - L) \frac{k'(k'-1)}{(n-2L)(n-2L-1)} \ge (m - L) \frac{(k/2)(k/2-1)}{n^2} \geq \Omega(m - L).
\]
Since $\E_{Y}[L] = m/2$, we get $\E_{Y}[m - L] = m/2 \ge \Omega(k)$. Thus, $\E_{P, Y, \mathcal{E}}[Z] = \Omega(k)$ and $\E_{P, Y, \mathcal{E}}[X] \ge \Omega\left(e^{-\eps} \cdot k\right)$. Hence, $\E_{P, Y, \mathcal{E}}[|I_S|] \le k\left(1 - \Omega\Paren{e^{-\eps}}\right)$. 
\end{proof}


\subsection{Triangle-Transversal}

Finally, we argue that, for some problem and natural representation, it is \emph{impossible} to achieve good privacy-approximation tradeoff. In particular, we consider the Triangle-Transversal problem.
For such deletion problems, a natural representation is to generalize the Permutation decoder from \Cref{def:perm-rep} for Vertex Cover. In particular, the encoder outputs a permutation $\pi$ of $V$ and the decoder deletes the earliest vertex of each prohibited structure (triangles), as formalized below.

\begin{definition}[Permutation Representation for Triangle-Transversal] \label{def:perm-rep-triangle}
In the \emph{permutation representation}, the encoder outputs a permutation $\pi$ of the vertex set $V$ as the implicit representation. The decoder $\cD(G, \pi)$ operates by selecting the earliest element in $\pi$ for each triangle $\{u, v, w\} \subseteq V$ (where $\{u, v\}, \{v, w\}, \{u, w\} \in E$), forming the set $S_\pi(E) \coloneqq \{\min_\pi(\{u, v, w\}) \mid \{u, v, w\} \subseteq V \text{ where } \{u, v\}, \{v, w\}, \{u, w\} \in E\}$. 
\end{definition}

We show that, unlike Vertex Cover (\Cref{thm:vc-main}), this representation inherently prevents any strong privacy-approximation tradeoff:

\begin{theorem} \label{thm:sep-tt}
Any $\eps$-DP implicit representation scheme under the permutation representation (\Cref{def:perm-rep-triangle}) for Triangle-Transversal has approximation ratio at least $ \Omega(n \cdot e^{-\eps})$.
\end{theorem}

In other words, to obtain even constant approximation, one needs $\eps$ to be $\Omega(\ln n)$, meanwhile all of our algorithms can achieve constant approximation when $\eps$ is independent of $N$.

\begin{proof}[Proof of \Cref{thm:sep-tt}]
Let $n = 2n_0$. 
Partition $V$ into two disjoint subsets $V_0$ and $V_1$ of size $n_0$. Let $G$ be the complete bipartite graph $K_{n_0, n_0}$ between $V_0$ and $V_1$. 

Consider any $\eps$-DP implicit representation scheme $(\cE, \cD)$ under the permutation representation (\Cref{def:perm-rep-triangle}).
We start by proving the following claim:
\begin{claim} \label{claim:bad-edge}
There exist vertices $x, y$ on the same side $V_i$ such that
\begin{align*}
\kappa(x, y) := \E_{\pi \gets \cE(G)}[|\{w \in V_{1 - i} \mid w <_{\pi} x, y\}|] \geq \Omega(n).
\end{align*}
\end{claim}

\begin{proof}[Proof of \Cref{claim:bad-edge}]
Define $\kappa_{\pi}(x, y) = |\{w \in V_{1 - i} \mid w <_{\pi} x, y\}|$.
Taking the average of $\kappa(x, y)$ across all $x, y$, we have
\begin{align} \label{eq:exp-pair-expanded}
\E_{i \sim \{0,1\}, \{x, y\} \sim \binom{V_i}{2}}[\kappa(x, y)] = \E_\pi \left[\E_{i \sim \{0, 1\}, \{x, y\} \sim \binom{V_i}{2}}[\kappa_{\pi}(x, y)]\right].
\end{align}
To bound the inner expectation on the RHS, we further rearrange this expectation as follows:
\begin{align*}
&\E_{i \sim \{0, 1\}, \{x, y\} \sim \binom{V_i}{2}}[\kappa_{\pi}(x, y)] \\ &= \E_{i \sim \{0, 1\}, \{x, y\} \sim \binom{V_i}{2}}\left[\sum_{w \in V_{1-i}} \ind{w <_\pi x, y}\right] \\
&= n_0 \cdot \E_{i \sim \{0, 1\}, \{x, y\} \sim \binom{V_i}{2}, w \sim V_{1-i}}\left[\ind{w <_\pi x, y}\right] \\
&= n_0 \cdot \E_{i \sim \{0, 1\}, \{x, y\} \sim \binom{V_i}{2}, \{w, t\} \sim \binom{V_{1-i}}{2}}\left[\ind{w <_\pi x, y}\right] \\
&= \frac{n_0}{4} \cdot \E_{\{x, y\} \sim \binom{V_0}{2}, \{w, t\} \sim \binom{V_{1}}{2}}\left[\ind{w <_\pi x, y} + \ind{t <_\pi x, y} + \ind{x <_\pi w, t} + \ind{y <_\pi w, t}\right]
\end{align*}
where, in the second-to-last equality, we simply introduce $t$ which does not affect the distribution of $w$, and, in the last equality, we use the symmetry between $\{x, y\}$ and $\{w, t\}$.

Finally, observe that, at least one of $\ind{w <_\pi x, y}, \ind{t <_\pi x, y}, \ind{x <_\pi w, t}, \ind{y <_\pi w, t}$ must be true because one of $x, y, w, t$ must appear first in $\pi$. Hence, we have
\begin{align*}
\E_{i \sim \{0, 1\}, \{x, y\} \sim \binom{V_i}{2}}[\kappa_{\pi}(x, y)] \geq \frac{n_0}{4} = \frac{n}{8}.
\end{align*}
Plugging this into \eqref{eq:exp-pair-expanded} yields $\E_{i, x, y}[\kappa(x, y)] \geq \Omega(n)$. Thus, there exist some $x, y$ on the same side $V_i$ such that $\kappa(x, y) \geq \Omega(n)$.
\end{proof}

Let $x, y \in V_i$ be as guaranteed in \Cref{claim:bad-edge}. Consider the graph $G' = (V, E')$ which is the same as $G$ except with an additional edge $(x, y)$. We have $\opt(G') = 1$ since we can remove either $x$ or $y$, which makes the graph triangle-free. 

Meanwhile, for a given representation $\pi$, the decoder will remove all $w \in V_{1 - i}$ such that $w <_{\pi} x, y$ because $\{w, x, y\}$ forms a triangle. This implies that the expected solution size produced by $(\cE, \cD)$ on $G'$ is at least
\begin{align*}
\E_{\pi \gets \cE(G')}[|\{w \in V_{1 - i} \mid w <_{\pi} x, y\}|] \geq e^{-\eps} \cdot \E_{\pi \gets \cE(G)}[|\{w \in V_{1 - i} \mid w <_{\pi} x, y\}|] \geq \Omega(n \cdot e^{-\eps}),
\end{align*}
where the first inequality is due to $\eps$-DP and the second is due to \Cref{claim:bad-edge}.
\end{proof}

\section{Conclusion and Open Questions}
\label{sec:conclusion}

In this paper, we advanced the study of combinatorial optimization under DP by generalizing the implicit representation model and allowing the encoder to run in FPT time. This relaxed computational model allowed us to bypass polynomial-time inapproximability barriers and achieve $(1 \pm \tO(1/\eps))$-approximations for fundamental graph problems such as Vertex Cover, $d$-Hitting Set, and maximum induced subgraphs on sparse graphs. We also formalized the encoder-decoder framework and established representation-independent lower bounds under Gap-ETH with advice, ruling out any hope of obtaining good approximation ratios when $\eps$ is too small.

Our work leaves several intriguing avenues for future research:
\begin{enumerate}[nosep]
    \item \textbf{Better Approximations.}
    An obvious open question is to improve the privacy-approximation tradeoff of our algorithms in \Cref{sec:alg}. There are two directions here:
    \begin{itemize}[nosep]
    \item \textbf{Current Representation:} First is to keep the same representation as in \Cref{sec:alg}. In this direction, the remaining gap in the approximation ratios between \Cref{sec:alg} and the lower bounds in \Cref{sec:rep-lower-bounds} is in the dependency on $\eps$. The former has $\tO(1/\eps)$ dependency whereas the latter has $\Omega(e^{-\eps})$ dependency. Closing this gap is an intriguing question.
    \item \textbf{New Representation:} Another potential direction is to use completely different representations. Our representation-independent lower bounds in \Cref{sec:lower-bounds} still apply here. However, these only hold for small $\eps$. When $\eps$ is a large constant, our lower bounds do not even rule out \emph{exact} algorithms. Giving such a scheme, or showing that this is impossible, would be a fundamental contribution to the model.
    \end{itemize}
    \item \textbf{Expanding the Implicit Representation Framework for FPT Algorithms.} While we developed techniques for several canonical problems and obtained some initial results, the FPT literature is rich and full of many innovative techniques, including those beyond graph algorithms. (See, e.g.,  \cite{DowneyF13,CyganFKLMPPS15}.) Extending the framework to other FPT techniques remains an exciting direction.
    \item \textbf{Beyond FPT vs Polynomial-Time.} While we restrict our decoder to run in polynomial time and let our encoder run in FPT time, it makes sense, in principle, to consider the implicit representation framework under any setting where the decoder is restricted to be ``simpler'' than the encoder. For example, we may allow the encoder to run in polynomial time and the decoder to run only in sublinear time. Another natural setting is to restrict the decoder to be ``local'' in the sense that it must decode whether a vertex is in the solution by just looking at its neighbors and the representation. (This is satisfied by the permutation representation.)
    \item \textbf{Approximate-DP.} While we keep the scope of our work to just \emph{pure-DP} (i.e., $\eps$-DP), it is a natural direction to also consider \emph{approximate-DP} (i.e., $(\eps, \delta)$-DP) as well. We remark that some of the lower bounds (such as \Cref{thm:rep-lb-noisy,thm:sep-tt}) can be easily extended to the approximate-DP where $\delta$ is a sufficiently small constant (depending on $\eps$). However, some other lower bounds, such as the representation-independent lower bounds (\Cref{thm:main-rep-ind-lb}), do not transfer directly to the approximate-DP setting. Hence, a strong separation between the two settings remains a possibility.
\end{enumerate}

\paragraph{Acknowledgment.} We thank Badih Ghazi for helpful discussions in the early stages of this work. Pasin would also like to thank Jittat Fakcharoenphol, Bingkai Lin and Vorapong Suppakitpaisarn for insightful discussions on related problems.

\paragraph{AI Disclosure.} All core ideas of the algorithms and proofs in this work were discovered by the authors. We used Gemini to help flesh out proof sketches to drafts of full proofs, and to help with drafting the introduction and conclusion. We have verified and heavily edited these drafts. We take full responsibility for all claims presented in this paper.

\bibliographystyle{alpha}
\bibliography{ref}

\end{document}